\documentclass[journal]{IEEEtran}
\usepackage{xcolor}
\usepackage{cite}
\usepackage{textcomp}
\usepackage{lipsum}
\usepackage{mwe}
\usepackage{rotating}
\usepackage{lineno}
\usepackage{amssymb}
\usepackage{amsmath}
\usepackage{amsthm}
\usepackage{epstopdf}
\usepackage{graphicx}
\usepackage{multirow}
\usepackage{tabls}
\usepackage{ctable}
\usepackage{algorithm}
\usepackage{algpseudocode}
\usepackage[utf8]{inputenc}
\usepackage[hidelinks]{hyperref}
\usepackage[english]{babel}
\DeclareGraphicsExtensions{.eps}
\usepackage[list=true,listformat=simple]{subcaption}
\usepackage{rotating}
\usepackage{lineno}
\usepackage{amssymb}
\usepackage{amsmath}
\usepackage{amsthm}
\usepackage{epstopdf}
\usepackage{graphicx}
\usepackage{multirow}
\usepackage{tabls}
\usepackage{ctable}
\usepackage{algorithm}
\usepackage{algpseudocode}
\usepackage[utf8]{inputenc}
\usepackage[hidelinks]{hyperref}
\usepackage[english]{babel}
\DeclareGraphicsExtensions{.eps}

\newcommand{\ttot}{T_{\scriptscriptstyle Total}}

\newcommand{\pfcccs}{P_{\scriptstyle f}^{\scriptscriptstyle CCCS}}
\newcommand{\pdcccs}{P_{\scriptstyle d}^{\scriptscriptstyle CCCS}}

\newcommand{\pfcccsdet}{P_{\scriptstyle f}^{\scriptscriptstyle CCCS, det}}
\newcommand{\pdcccsdet}{P_{\scriptstyle d}^{\scriptscriptstyle CCCS, det}}

\newcommand{\rcccs}{R_{\scriptscriptstyle CCCS}}

\newcommand{\ecccs}{E_{\scriptscriptstyle CCCS}}

\newcommand{\eecccs}{EE_{\scriptscriptstyle CCCS}}

\newcommand{\trcccs}{\tilde{R}_{\scriptscriptstyle cccs}}
\newcommand{\trcccsdet}{\tilde{R}_{\scriptscriptstyle cccs, det}}

\newcommand{\tecccs}{\tilde{E}_{\scriptscriptstyle cccs}}
\newcommand{\tecccsdet}{\tilde{E}_{\scriptscriptstyle cccs, det}}

\newcommand{\teecccs}{\tilde{EE}_{\scriptscriptstyle cccs}}
\newcommand{\teecccsdet}{\tilde{EE}_{\scriptscriptstyle cccs, det}}

\newcommand{\meuvec}{\boldsymbol{\mu}}

\newcommand{\meas}{\boldsymbol{\phi}}
\newcommand{\wvec}{\mathbf{w}}
\newcommand{\xvec}{\mathbf{x}}
\newcommand{\xmat}{\mathbf{X}}
\newcommand{\yvec}{\mathbf{y}}
\newcommand{\zvec}{\mathbf{z}}
\newcommand{\ymat}{\mathbf{Y}}
\newcommand{\zmat}{\mathbf{Z}}

\newcommand{\varn}{\sigma_w^2}
\newcommand{\vars}{\sigma_x^2}
\newcommand{\vark}{\sigma_k^2}
\newcommand{\varzero}{\sigma_0^2}
\newcommand{\varone}{\sigma_1^2}

\newcommand{\eyep}{\mathrm{I}_{\scriptscriptstyle P}}
\newcommand{\phn}{\pi_0}
\newcommand{\pho}{\pi_1}
\newcommand{\phat}{\widehat{\mathbf{P}}}
\newcommand{\comp}{\mathtt{c}}
\newcommand{\nbnd}{N_{\scriptscriptstyle \text{$\max$}}}
\newcommand{\compbnd}{\mathtt{c}_{\scriptscriptstyle \text{$\max$}}}
\newcommand{\compupbnd}{\mathtt{c}_{\scriptscriptstyle \text{$UB$}}}

\newtheorem{theorem}{Theorem}

\title{Energy Efficiency Analysis of Collaborative Compressive Sensing Scheme in Cognitive Radio Networks}

\author{Rajalekshmi Kishore,~\IEEEmembership{Student Member,~IEEE}, Sanjeev~Gurugopinath,~\IEEEmembership{Member,~IEEE},\\Sami Muhaidat, \IEEEmembership{Senior Member,~IEEE}, Paschalis C.~Sofotasios,~\IEEEmembership{Senior Member,~IEEE},\\Mehrdad~Dianati,~\IEEEmembership{Senior Member,~IEEE}, and Naofal Al-Dhahir,~\IEEEmembership{Fellow,~IEEE}

	\thanks{This work has appeared in part in \cite{Kishore_Globecom_2018}.} 
	
	\thanks{R. Kishore is with the Department of Electrical and Electronics Engineering, BITS Pilani, K.~K.~Birla Goa Campus, Goa 403726, India (email: {\rm lekshminair2k@yahoo.com}).}

	\thanks{S. Gurugopinath is with the Department of Electronics and Communication Engineering, PES University, Bengaluru 560085, India, (email: {\rm sanjeevg@pes.edu}).}

	\thanks{S.  Muhaidat and P. C. Sofotasios are with the Department of Electrical and Computer Engineering, Khalifa University of Science and Technology, Abu Dhabi, UAE  (emails: {\rm \{muhaidat, p.sofotasios\}@ieee.org}).}
	
	\thanks{Mehrdad~Dianati is with the Warwick Manufacturing Group, University of Warwick, United Kingdom. (email: {\rm   M.Dianati@warwick.ac.uk }).}

	\thanks{N. Al-Dhahir is with the Department of Electrical Engineering, University of Texas at Dallas, TX 75080 Dallas, USA (e-mail: {\rm aldhahir@utdallas.edu}).}	
}
 
\begin{document}

 \maketitle
 
\begin{abstract}
In this paper, we investigate the energy efficiency of conventional collaborative compressive sensing (CCCS) scheme, focusing on balancing the tradeoff between energy efficiency and detection accuracy in cognitive radio environment. In particular, we derive the achievable throughput, energy consumption and energy efficiency of the CCCS scheme, and formulate an optimization problem to determine the optimal values of parameters which maximize the energy efficiency of the CCCS scheme. The maximization of energy efficiency is proposed as a multi-variable, non-convex optimization problem, and we provide approximations to reduce it to a convex optimization problem. We highlight that errors due to these approximations are negligible. Later, we analytically characterize the tradeoff between dimensionality reduction and collaborative sensing performance of the CCCS scheme -- the implicit tradeoff between energy saving and detection accuracy, and show that the loss due to compression can be recovered through collaboration which improves the overall energy efficiency.
\end{abstract}

\begin{IEEEkeywords}
Achievable throughput, collaborative compressive sensing, energy consumption, energy efficiency, spectrum sensing.
\end{IEEEkeywords}

\IEEEpeerreviewmaketitle

\section{Introduction}
\label{sec:introduction}
With growing concern about environmental issues and an emerging green communications paradigm (\cite{Mowla_IEEETGCN_2017}, \cite{Huang_IEEECST_2015}) in wireless communications, the design of cognitive radio (CR) networks (CRNs) have to be considered from the energy efficiency perspective (\cite{Zheng_IEEETGCN_2017}, \cite{Yousefvand_IEEETGCN_2017}). A fundamental feature of a CR is spectrum sensing \cite{Yucek_IEEECST_2009}, which is typically carried out by the CR users or secondary users (SU) to find the unused licensed resources for implementing a CRN or a secondary network. 

It is well-understood that larger the bandwidth of the licensed or primary user (PU) spectrum, the SUs will have more transmission opportunity for communication. Towards this end, wideband spectrum sensing (WSS) \cite{Bruno_IEEETC_2018},\cite{Cao_IEEETSP_2018} has attracted considerable research attention, to design efficient algorithms for detecting multiple bands simultaneously. Typically, the duration of a spectrum sensing slot includes two phases, namely, the sensing phase and the data transmission phase. If the sensing phase is not optimally designed, the energy consumption of SUs increases. Such a design problem is of primary importance for WSS \cite{Ali_IEEETVT_2017}. The energy consumption for spectrum sensing, mainly
caused by the analog-to-digital converter (ADC), is proportional to the sensing time and the sampling rate (\cite{Xiong_IEEETVT_2017}, \cite{Zhao_IEEETMC_2017}). However, it has been observed that at a given time instant, only a small number of frequency bins (channels) across the entire bandwidth are occupied by PUs. In other words, the occupancy of the PU network over a wideband is sparse in the frequency-domain. Such inherent sparsity of the spectrum is taken as an advantage in compressed sensing (CS)-based approaches, which was originally envisioned to reduce the sampling rate below the Nyquist rate \cite{Sharma_IEEE_2016}. Based on this key observation, the authors in \cite{SALAHDINE_ADhoc_2016} present an extensive survey on compressive sensing techniques and discuss about the classification of these techniques, their potential applications and metrics to optimally design and evaluate their performances in the context of CRNs. To summarize, CS, when compared to the conventional WSS, reduces the sampling rate to below Nyquist rate \cite{Ma_IEEEVT_2017}, which in turn reduces the sensing time, favoring considerable saving in energy consumption. For this reason, the CS-based spectrum sensing methods have been proposed for improving the energy efficiency \cite{Arienzo_IEEETGCN_2017} in CRNs.

Despite its attractiveness as an energy efficient sensing technique, CS suffers from a few major drawbacks which limit its applicability in practice. A CS based sensing scheme incurs a considerable performance loss due to compression when compared to the conventional sensing scheme, while detecting non-sparse signals. This performance loss is characterized in terms of the probabilities of false-alarm and signal detection. Recently, the authors in \cite {varshney_IEEE_2017} proposed a collaborative compressive detection framework, in which group of spatially distributed nodes sense the presence of phenomenon independently, and send a compressed summary of observations to a fusion center (FC) where a global decision is made about the presence or absence of the phenomenon. This technique was designed to compensate for the performance loss due to compression, and it was shown that the amount of loss can be improved and recovered through collaborative detection. In particular, it was shown that as the the degree of compression is decreased (keeping number of collaborating nodes fixed), or as the number of collaborating nodes is increased (keeping the degree of compression fixed), the overall probability of error in detection can be made arbitrarily small. However, the study in \cite{varshney_IEEE_2017} never addressed energy efficiency and was restricted to the detection performance of the collaborative compressive detection scheme, in a non-CR context.

In this work, we have shown that a similar trend observed in \cite{varshney_IEEE_2017} can be seen in CRNs, with energy efficiency as a metric. In particular, we derive the expressions for the average energy consumption and the average achievable throughput of a conventional collaborative compressive sensing (CCCS) scheme. Next, we derive an expression for the energy efficiency of CCCS, and formulate an optimization problem that maximizes the energy efficiency, subjected to constraints on probability of detection and probability of false-alarm. We provide some approximations to reduce the proposed non-convex optimization problem to a convex optimization problem. Later, we establish that these approximations are sufficiently
accurate, and result only in an insignificant performance loss. The motivation to consider the proposed CCCS is threefold. First, it reduces the sampling rate below the Nyquist rate, which results in a shorter sensing duration and much lesser energy consumption. Secondly, by exploiting the collaboration between the sensors, the achievable detection performance can be maintained to a target limit. Finally, since it promotes energy saving and ensures a desirable detection performance, the energy efficiency is guaranteed. In the process of determining optimal system parameters such as the degree of compression (or the compression ratio) and number of collaborative nodes, we seek the answer to the following question: For a given compression ratio, what would be the minimum number of collaborative nodes required to maximize the energy efficiency of the CRN$?$\footnote{A related question would be that given a number of collaborative nodes $N$, what is the maximum allowable degree of compression, such that the energy efficiency of the network is maximized?}

On a related note, the energy efficiency using compressed sensing in wideband CRNs was studied in \cite{Zhao_eurosipSpringer_2016}, where the authors show that by optimizing the sampling rate, energy efficiency of the network can be maximized. It was also shown that as the sparsity of the wideband spectrum increases (that is, as the associated vector becomes more and more sparse), the energy consumption decreases, and the energy efficiency increases. But the analysis in \cite{Zhao_eurosipSpringer_2016} was restricted to strictly sparse signals. However, in this work, we have considered the utility of both compressed sensing and collaborative sensing to guarantee dimensionality reduction and detection performance, respectively, that yields improvement in energy efficiency to a greater extent. Moreover, our approach is also applicable to non-sparse signals. To the best of our knowledge, such an analysis on energy efficiency for the CCCS scheme has not been considered earlier in the literature.

The main contributions of this paper are as follows.
\begin{itemize}
	\item Energy efficiency of the CCCS scheme for CRNs is studied, in terms of the average achievable throughput and the average energy consumption in the network.
	\item Maximization of the energy efficiency is posed as a non-convex optimization problem, to find the number of sensors required for collaboration (or the degree of compression), that satisfies a given constraints on probability of false-alarm and probability of detection.
	\item A study on the effect of reducing the number of samples due to CS, and its impact on the energy efficiency is carried out, considering the random and deterministic PU signal models. In both cases, we show that the energy efficiency is improved by either decreasing the compression ratio, or by increasing the number of collaborative nodes.
	\item Through numerical results, we compare the performances of the conventional collaborative sensing (CCS) and CCCS schemes in terms of the energy efficiency, and highlight the regimes where CCCS outperforms the CCS scheme. Such an improvement in energy efficiency of the CCCS scheme is shown to be due to a significant amount of saving in the energy consumption, with a relatively insignificant performance loss due to detection accuracy, in comparison to the CCS scheme.
\end{itemize}

The remainder of this paper is organized as follows. We propose the system model for CCCS scheme and review the CCCS and CCS schemes for random PU signal case in Sec.~\ref{SecSysModel}.  The optimization problem to maximize the energy efficiency of the CCCS scheme is proposed in Sec.~\ref{SecCCCEEProblem}, and associated approximations, reformulation and detailed analysis are provided in Sec.~\ref{SecCCCSEEApprox}. A similar energy efficiency formulation, approximations, and analysis for a deterministic PU signal is presented in Secs.~\ref{EEforDetersignal}. Numerical results and discussion on performance comparison are presented in Sec.~\ref{SecResults} and concluding remarks are provided in Sec.~\ref{SecConc}.

\section{System Model} \label{SecSysModel}

\begin{figure*} \vspace*{4pt}
	\centering
	\begin{minipage}{.5\textwidth}
		\centering
		\includegraphics[width=1\linewidth]{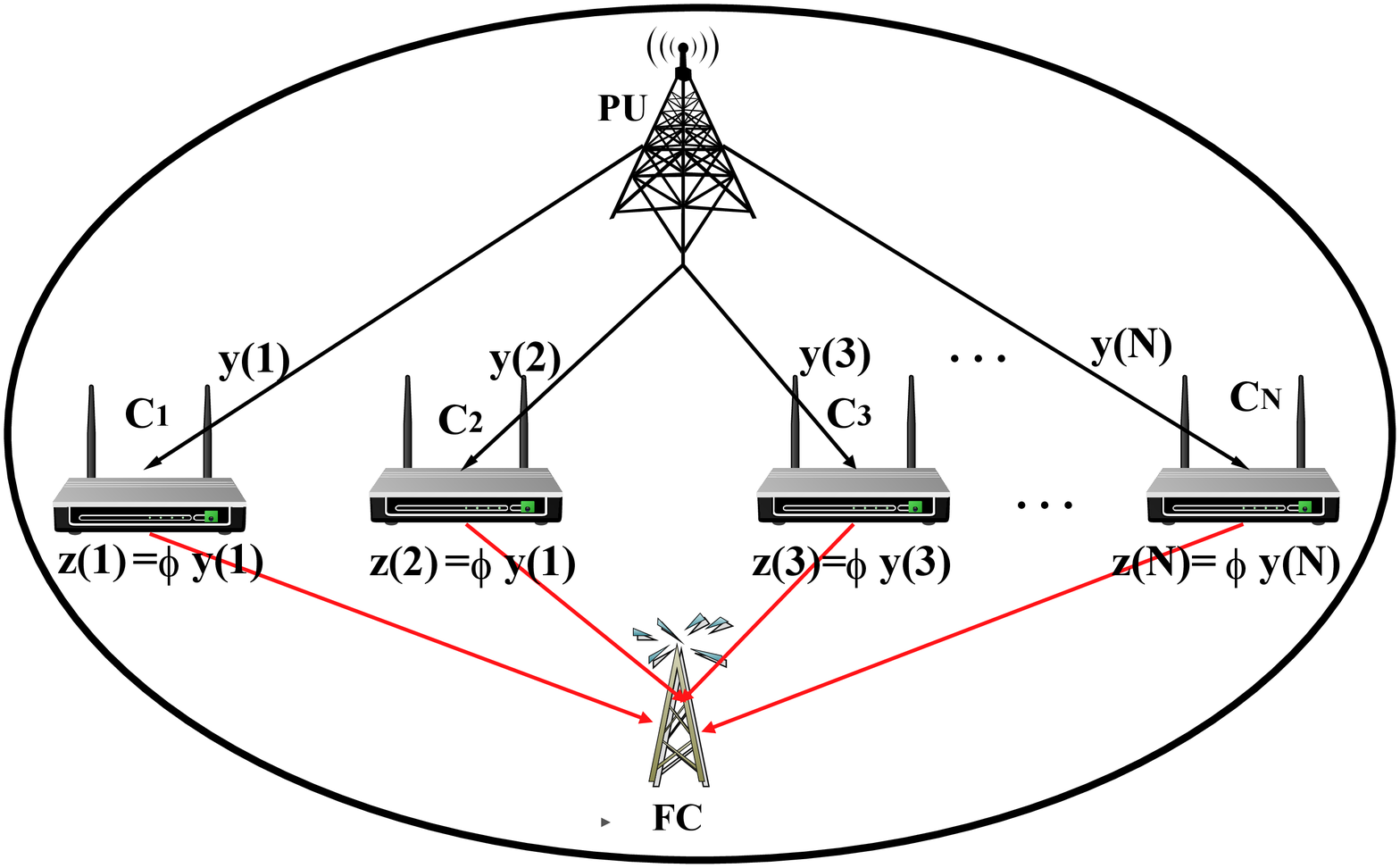}
		\centering
	\end{minipage}%
	\begin{minipage}{.5\textwidth}
		\centering
		\includegraphics[width=1\linewidth]{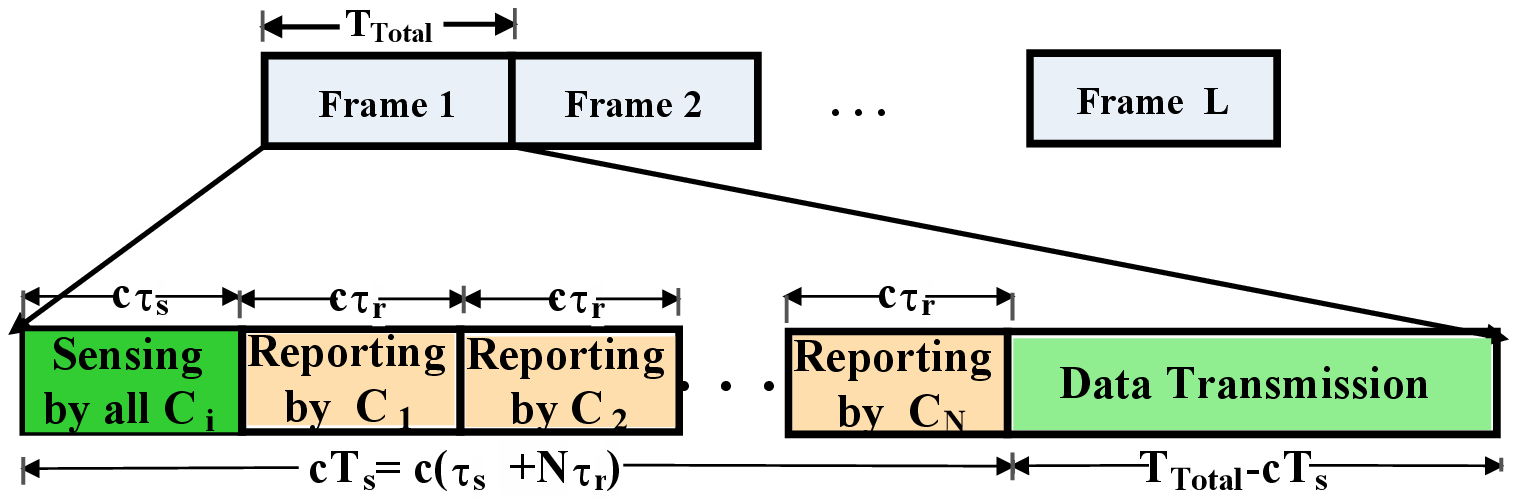}
		\centering
	\end{minipage}
	\captionof{figure}{(a) System model for collaborative conventional compressive sensing (CCCS) scheme ~ (b) Time slot structure for CCCS scheme.}
	\label{fig:test2}
\end{figure*}
We first describe the conventional cooperative sensing (CCS) framework. Consider a CRN -- as depicted in Fig.~\ref{fig:test2}(a) -- with $N$ CR nodes denoted by $C_1,\ldots,C_N$ that record $P$ observations each from a licensed band owned by a primary user (PU). These nodes forward their observation vectors over a lossless link to a fusion center (FC), where they are fused to make an overall decision on the availability of the primary spectrum. The hypothesis testing problem governing this scenario can be written as
\begin{align}
 & \mathcal{H}_{0}: \yvec(n) = \wvec(n) \nonumber \\
 & \mathcal{H}_{1}: \yvec(n) = \xvec(n) + \wvec(n), ~~ n=1,\ldots,N, \label{MainHIProb}
\end{align}
where $\wvec(n)$ represents the $P \times 1$ noise vector, and $\xvec(n)$ represents a $P \times 1$ primary signal vector, whose entries are assumed to be i.i.d.~Gaussian random variables with zero mean and variance $\varn$ and $\vars$, respectively. That is, if $\mathcal{N}(\boldsymbol{\meuvec}, \boldsymbol{\Sigma})$ denotes a Gaussian random vector with mean vector $\boldsymbol{\meuvec}$ and covariance matrix $\boldsymbol{\Sigma}$, then $\wvec(n) \sim \mathcal{N}(\boldsymbol{0},\varn \eyep)$, and $\xvec(n) \sim \mathcal{N}(\boldsymbol{0},\vars \eyep)$, where $\eyep$ is a $P \times P$ identity matrix.

Next, we focus on the conventional collaborative compressive sensing (CCCS) framework. Here, instead of $P \times 1$ vector $\yvec(n)$, each node sends an $M \times 1$ compressed vector $\zvec(n)$ to the FC, with $M < P$. The collection of these $M$-length universally sampled observations is given by $\{\zvec(n) = \meas \yvec(n), n=1,\ldots,N\}$, where $\meas$ is an $M \times P$ fat compression matrix, which is assumed to be the same across all nodes. With this setup, the problem in \eqref{MainHIProb} reduces to
\begin{align}
& \mathcal{H}_0: \zvec(n) = \meas \wvec(n) \nonumber \\
& \mathcal{H}_1: \zvec(n) = \meas(\xvec(n) +  \wvec(n)), ~~ n=1,\ldots,N, \label{CompHIProb}
\end{align}
The FC receives the observation matrix $\zmat = [\zvec(1) \cdots \zvec(N)]$, and makes a decision on the availability of the primary spectrum, by employing the likelihood ratio test (LRT), which is Neyman-Pearson optimal. The LRT at the FC, with a detection threshold $\lambda_L$, is given as
\begin{align}
 \prod_{n=1}^N \frac{f(\zvec(n); \mathcal{H}_1)}{f(\zvec(n); \mathcal{H}_0)} \overset{\mathcal{H}_1}{\underset{\mathcal{H}_0}{\gtrless}} \lambda_L, \label{LRTeqn}
\end{align}
where $f(\zvec(n); \mathcal{H}_0)$ and $f(\zvec(n); \mathcal{H}_1)$ represent the PDF of $\zvec(n)$ under $\mathcal{H}_0$ and $\mathcal{H}_1$, and are respectively given by
\begin{align}
& f(\zvec(n); \mathcal{H}_0) = \dfrac{\exp\left(-\frac{\zvec^T(n) (\varn \meas \meas^T)^{-1} \zvec(n)}{2}\right)}{(2\pi)^{M/2} |\varn \meas \meas^T|^{1/2}} ~~
\vspace*{2cm} 
\\ 
& f(\zvec(n); \mathcal{H}_1) = \dfrac{\exp\left(-\frac{\zvec^T(n) ((\vars+\varn) \meas \meas^T)^{-1} \zvec(n)}{2}\right)}{(2\pi)^{M/2} |(\vars+\varn) \meas \meas^T|^{1/2}}.
\end{align}
Substituting in \eqref{LRTeqn} and simplifying, yields
\begin{align}
&\left[\frac{|\varn \meas \meas^T|^{1/2}}{|(\vars+\varn) \meas \meas^T|^{1/2}}\right]^N \hspace{-0.2cm} \exp \hspace{-0.1cm} \left[ \hspace{-0.075cm} - \hspace{-0.1cm} \sum_{n=1}^N \left(\frac{\zvec^T(n) (\meas \meas^T)^{-1} \zvec(n)}{2(\vars+\varn)}\right. \right. \nonumber \\
& ~~~~~~~~~~~~~~~~~~~~~~~~~ \left. \left. - \frac{\zvec^T(n) (\meas \meas^T)^{-1} \zvec(n)}{2\varn} \right)\right] \overset{\mathcal{H}_1}{\underset{\mathcal{H}_0}{\gtrless}} \lambda_L. \label{LRTeqn1}
\end{align}
Recalling that $\zvec(n) = \meas \yvec(n)$, it is easy to see that the above test reduces to the form
\begin{align}
T(\ymat) \triangleq \sum \limits_{n=1}^{N} \yvec^T(n) \meas^T (\meas \meas^T)^{-1} \meas \yvec(n) \overset{\mathcal{H}_1}{\underset{\mathcal{H}_0}{\gtrless}} \lambda, \label{SuffTest}
\end{align}
 where $\lambda \triangleq \log \left \{\left [ \frac{|\sigma_x^2+\sigma_w^2|}{|\sigma_w^2|} \right ]^{N/2} \lambda_L \right \}\left \{ \frac{2\sigma_w^2(\sigma_w^2+\sigma_x^2)}{\sigma_x^2} \right \}$ is the detection threshold, which is chosen based on the Neyman-Pearson criterion. To simplify performance characterization of the above test in \eqref{SuffTest}, we assume that the linear mapping $\meas$ satisfies the $\epsilon$-embedding property, as considered in \cite{varshney_IEEE_2017}. However, designing such a $\meas$ that satisfies the $\epsilon$-embedding property is beyond the scope of the current study. 

Let $\gamma \triangleq \frac{\vars}{\varn}$ denote the average received SNR at a CR node, and $\phat \triangleq \meas^T(\meas \meas^T)^{-1}\meas$ the  projection matrix on the row space of $\meas$. Following the central limit theorem for large values of the product $NM$, it can be shown that the test statistic under both $\mathcal{H}_0$ and $\mathcal{H}_1$ is distributed as
\begin{eqnarray} \label{teststatdet}
\frac{T (\ymat)}{\vark} \overset{NM \rightarrow \infty}{\sim} \left\lbrace \begin{array}{cl} \mathcal{N}(NM,2NM), & \hbox{under}\ \mathcal{H}_{0} \\ \mathcal{N}(N M,2 N M), & \hbox{under}\ \mathcal{H}_{1} \end{array}\right. \hspace{-0.15cm}
 \end{eqnarray}
where $k=0,1$, that is, $\varzero=\varn$, and $\varone=\vars+\varn$. Let $\comp \triangleq \frac{M}{P} \in (0,1)$ denote the compression ratio. Based on \eqref{teststatdet}, the probability of false-alarm at the FC is given by
\begin{align}
& \pfcccs \triangleq P(T(\ymat) > \lambda |\mathcal{H}_{0}) = Q\left( \frac{\frac{\lambda}{P \varn}-\comp N}{\sqrt{2 \comp N P}} \right). \label{pfeqn}
\end{align}
Similarly, the probability of detection at the FC is given by
\begin{align}
&\pdcccs \hspace{-0.1cm} = \hspace{-0.1cm} P(T(\ymat)>\lambda|\mathcal{H}_{1}) = \hspace{-0.1cm} Q\left(\hspace{-0.1cm} \frac{\frac{\lambda}{P (\vars+\varn)}-\comp N}{\sqrt{2 \comp N P}} \hspace{-0.1cm} \right). \label{pdeqn}
\end{align}

Note that the expressions for $\pfcccs$ and $\pdcccs$ depend on the value of $\comp$, which dictates the loss in the detection accuracy due to the compressed measurements $\{\zvec(n), n=1,\ldots,N\}$. The time slot structure indicating the sensing, reporting and total duration for the CCCS scheme is as shown in Fig.~\ref{fig:test2}(b). Since the detection accuracy is also a function of $N$, it can be improved by increasing $N$. In other words, the loss in detection accuracy due to compression can be recovered by increasing the number of collaborative nodes, $N$. This observation is shown in Fig.~\ref{PdVsSNR}, where the variation of $\pdcccs$ across $\gamma$ is plotted, with $\pfcccs=0.1$, for different values of $\comp$ and $N$. The case of $\comp=1$ corresponds to Nyquist sampling, i.e., the CCS approach. As $\comp$ decreases, $\pdcccs$ decreases, which can be increased to a desired level by increasing $N$. Interestingly, as $N$ increases, even though the probability of detection -- and consequently, the achievable throughput of the secondary network -- increases, the total energy consumption in the secondary network also increases, thereby decreasing the energy efficiency. Towards this end, it is of  paramount importance to optimally determine  system parameters $\comp$ and $N$ in order to maximize the energy efficiency. In other words, we seek to answer the following question. Given a CR network with $N$ nodes, how small can the compression ratio $\comp$  be, such that the energy efficiency is maximized? To answer this question -- which is the main contribution of this paper, we next derive expressions for the average achievable throughput, average energy consumption and energy efficiency of the CR network, and formulate an optimization problem to maximize the energy efficiency.

\begin{figure}
	\centering
	\vspace{-0.5cm}
	\includegraphics[width=3.5in, height=2.5in]{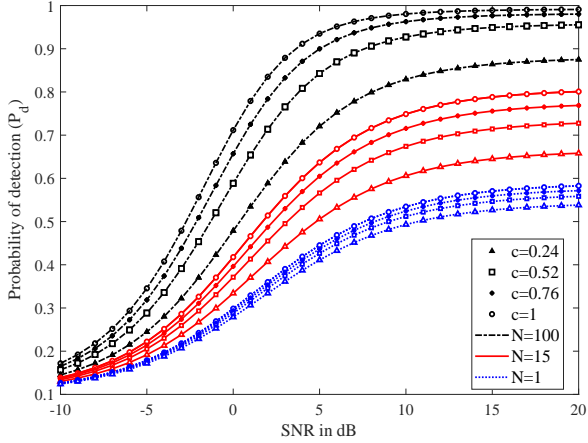} 
	\caption{Variation of probability of detection, $\pdcccs$, for different values of average SNR, $\gamma$. Probability of false-alarm, $\pfcccs=0.1$. Note that as $\comp$ decreases, $\pdcccs$ decreases. However, $\pdcccs$ can be increased to a desired level by increasing $N$.} \label{PdVsSNR}
\end{figure}

\section{Energy Efficiency and Problem Formulation} \label{SecCCCEEProblem}
In this section, our aim  is to find the optimal value of the compression ratio $\comp$, for a given $N$, such that the energy efficiency  is maximized. To this end, we first derive expressions for the average achievable throughput and the average energy consumption, and then derive the energy efficiency of the network. For the underlying CCCS, the average achievable throughput and the average energy consumption depend on the communication link between the PU node and the sensing nodes, and can be calculated based on the following scenarios, where  $\phn$ and $\pho$ denote the prior probability that the channel is vacant and occupied, respectively.

\begin{table*}[t]
	\centering
	\caption{Achievable throughput and energy consumption by the CR network employing CCCS, for  scenarios \textbf{S1}-\textbf{S4}.}
	\label{TputEconTable}
	\begin{tabular}{p{2cm}p{4cm}p{6cm}p{3cm}}\cline{1-4}
		\hline\noalign{\smallskip}\\
		\textbf{Scenario} & \textbf{Probability} & \textbf{Energy Consumed (J)}  & \textbf{Achievable Throughput (bits/Hz)}\\
		\noalign{\smallskip}\hline\noalign{\smallskip}\noalign{\smallskip}
		\textbf{S1} &	$\pho$ $\pdcccs$  & $NP_{s}\comp\tau_{s}$+$N P_{t}\comp \tau_{r}$  & 0  \\ \vspace{0.15cm}\\ 
		\cline{1-4}\\
		\textbf{S2} &$\phn$ $\pfcccs$ & $NP_{s}\comp\tau_{s}$+$N P_{t}\comp \tau_{r}$ & $-\phi \mathcal{C} (\ttot -\comp T_{s})$ \\ \vspace{0.15cm}\\
		\cline{1-4} \\
		\textbf{S3} &$\pho$ $(1-\pdcccs)$ & $NP_{s}\comp\tau_{s}$+$N P_{t}\comp \tau_{r}+P_{t}(\ttot-\comp T_{s})$  & $\kappa_c \mathcal{C} (\ttot -\comp T_{s})$  \\ \vspace{0.15cm}\\
		\cline{1-4} \\
		\textbf{S4} &$\phn$ $(1-\pfcccs)$ & $NP_{s}\comp\tau_{s}$+$N P_{t}\comp \tau_{r}+P_{t}(\ttot-\comp T_{s})$ & $ \mathcal{C} (\ttot -\comp T_{s})$  \\ \vspace{0.15cm}\\
		
		\noalign{\smallskip}\hline
	\end{tabular}
\end{table*}

In a CR network with CCCS, the average achievable throughput and the average energy consumption depend on the communication link between the PU node and the sensing nodes, which can be calculated based on four scenarios denoted by \textbf{S1}-\textbf{S4}, detailed below.
\begin{enumerate}
\item[\textbf{S1}.] The first scenario corresponds to the case when the PU is present, and the FC correctly declares its presence, which occurs with probability $\pho \pdcccs$. Hence, the CR network throughput achieved is zero.
\item[\textbf{S2}.] The second scenario covers the case when PU is absent but incorrectly declared as present by the FC, which occurs with probability  $\phn \pfcccs$. Since the CR network misses a transmission opportunity in this case, the achievable throughput in this case is calculated as $-\phi \mathcal{C} (\ttot-\comp T_{s})$, where $T_{s}= (\tau_{s} + N \tau_{r})$, $\mathcal{C}$ is the capacity of the secondary link, and $\phi \in (0,1)$ is a suitably chosen penalty factor. For simplicity, $\phi$ can be considered to be zero.
\item[\textbf{S3}.] In the third scenario, FC makes an incorrect decision that the PU is absent, when it is actually present, which occurs with probability $\pho (1-\pdcccs)$. In this case, the CR network transmits and causes interference to the PU. Even with the interference to the PU, the CR communication achieves a partial throughput of $\kappa_c \mathcal{C} (\ttot- \comp T_{s})$ units, for some $\kappa_c \in [0,1)$. Additionally, we assume that the CR nodes are located far from the PU network, such that the interference term due to PU is negligible.
\item[\textbf{S4}.] The last scenario corresponds to the case when the PU is absent and the FC makes a correct decision, which occurs with probability $\phn (1-\pfcccs)$. In this case, the achievable throughput is maximum, and is given by $\mathcal{C}(\ttot- \comp T_{s})$ units.
\end{enumerate}

The achievable throughput, along with the energy consumed in each of the above scenarios are listed in Tab.~\ref{TputEconTable}, on the top of the next page, where $P_s$ and $P_t$ denote the power required for each SU node for sensing and data transmission, respectively. Considering all the above cases, the average throughput of the CCCS scheme is given by
\begin{align} 
  & \hspace{-0.15cm} \rcccs(\lambda,\comp,N) \hspace{-0.1cm} = \hspace{-0.1cm} \phn (1-\pfcccs)(\ttot-\comp T_{s})\mathcal{C} \nonumber \\
  &~~~~~~~~~~~~~~~~~~ + \kappa_{c}\mathcal{C}(\ttot-\comp T_{s}) \pho (1-\pdcccs) \nonumber\\ 
  &~~~~~~~~~~~~~~~~~~ - \phi \mathcal{C}(\ttot-\comp T_{s}) \phn \pfcccs. \label{RCCCSEqn}
\end{align}
Similarly, the average energy consumption of the CCCS scheme, as illustrated in Tab.~\ref{TputEconTable}, can be written as 
\begin{align} 
& \hspace{-0.15cm}\ecccs(\lambda,\comp,N) \hspace{-0.1cm} = \hspace{-0.1cm} \left(NP_{s}\mathtt{c}\tau_{s}+N P_{s}\comp \tau_{r}\right) \nonumber \\  &~~~+P_{t}(\ttot-\comp T_{s})\left(1-\pho \pdcccs- \phn \pfcccs\right).
\end{align}
Based on above, the energy efficiency, measured in (bits/Hz/J), of the underlying CR network is given by
\begin{equation}  \label{EEeqn}
\eecccs(\lambda, \comp, N) \triangleq \frac{ \rcccs(\lambda, \comp, N)}{\ecccs(\lambda, \comp, N)}.
\end{equation}

Recall that our goal here is to design $\lambda$ and $\comp$, for a given $N$, such that the energy efficiency $\eecccs(\lambda, \comp, N)$ is maximized, subject to constraints on the sensing errors. The optimization problems can be divided into two sub-categories, namely, optimizing $N$ for a given $\comp$, and optimizing $\comp$ for a given $N$. For a given $\comp$, the governing optimization problem is:
\begin{eqnarray}                
& \mathcal{OP}^{(N)}_{\scriptscriptstyle CCCS}:  \underset{N}{\max} ~~~ \eecccs(\comp, N) \nonumber \\
& ~~~~~~~~~~  s.t. ~~ \pfcccs \leq \overline{P}_{f}, \nonumber \\
& ~~~~~~~~~~~~~~~~ \pdcccs \geq \overline{P}_{d}, \label{OPwithN}
\end{eqnarray}
and for a given $N$, the governing optimization problem is given as
\begin{eqnarray}                
& \mathcal{OP}^{(\comp)}_{\scriptscriptstyle CCCS}:  \underset{\comp}{\max} ~~~ \eecccs(\comp, N) \nonumber \\
& ~~~~~~~~~~  s.t. ~~ \pfcccs \leq \overline{P}_{f}, \nonumber \\
& ~~~~~~~~~~~~~~~~ \pdcccs \geq \overline{P}_{d}. \label{OPwithc}
\end{eqnarray}
In the subsequent analysis, we assume that $0 \leq \overline{P}_{f} < \overline{P}_{d} \leq 1$. This is followed from the IEEE 802.22 standard \cite{Carl_IEEE_2009} requirements, where the lower bound on the probability of signal detection and upper bound on the probability of false-alarm are $0.9$ and $0.1$, respectively.

The problems given in \eqref{OPwithN} and \eqref{OPwithc} are hard to solve, because the expression for $\eecccs(\comp,N)$ calculated from \eqref{EEeqn} is lengthy. For the ease of analysis, we approximate the cost function in the above problems, and mention the conditions under which the problem can be reduced to a convex optimization problem. Later, in Sec.~\ref{SecResults}, we demonstrate that the corresponding error due to these approximations is negligible.

\subsection{Approximation, Reformulation and Analysis} \label{SecCCCSEEApprox}
In this section, we first provide an approximation of $\eecccs$ and reformulate the optimization problems \eqref{OPwithN} and \eqref{OPwithc}. On a general note, the apriori probability of channel availability should be large enough to maintain the detection accuracy. That is, we assume that $\phn (1-\pfcccs) > \pho (1-\pdcccs)$, which is justified in a typical CR scenario \cite{Peh_IEEE_2011},\cite{Gao_IEEE_2013}. Following this, the average throughput in \eqref{RCCCSEqn} can be approximated by the above inequalities and setting $\kappa_{c}=0$ as
\begin{eqnarray} \label{ApprxReqn}
\trcccs(\lambda, \comp, N) \approx \phn \mathcal{C} (\ttot-\comp T_{s})\left(1-(1+\phi)\pfcccs\right).
\end{eqnarray} 
Similarly $\ecccs(\lambda,\comp,N)$ can be approximated as
\begin{align} 
& \hspace{-0.15cm} \tecccs(\lambda, \comp, N) \hspace{-0.1cm} \approx \hspace{-0.1cm} \left(NP_{s}\comp\tau_{s}+N P_{s}\comp \tau_{r}\right) \nonumber \\
&~~~~~~~~~~ + P_{t}(\ttot-\comp T_{s}) \phn (1- \pfcccs). \label{ECCCSEqn}
\end{align}

Consequently, $\eecccs(\lambda,\comp,N)$ can be approximated as
\begin{eqnarray} 
\teecccs(\lambda,\comp,N) = \frac{\trcccs(\lambda,\comp,N)}{\tecccs(\lambda,\comp,N)}, \label{EECCCSapproxEqn}
\end{eqnarray} 
and the optimization problems $\mathcal{OP}^{(N)}_{\scriptscriptstyle CCCS}$ and $\mathcal{OP}^{(\comp)}_{\scriptscriptstyle CCCS}$ can be respectively reformulated as 
\begin{eqnarray}               
& \mathcal{OP}1^{(N)}_{\scriptscriptstyle CCCS}:  \underset{\lambda, N}{\max} ~~~ \teecccs(\lambda, \comp, N) =\frac{\trcccs(\lambda, \comp, N)}{\tecccs(\lambda, \comp, N)} \nonumber \\
& \hspace{-0.7cm}  s.t. ~~~~ \pfcccs \leq \overline{P}_{f}, \nonumber \\
& \hspace{-0.7cm} ~~~~~~~~ \pdcccs \geq \overline{P}_{d}, \label{OP2withN}
\end{eqnarray}
and
\begin{eqnarray}               
& \mathcal{OP}1^{(\comp)}_{\scriptscriptstyle CCCS}:  \underset{\lambda, \comp}{\max} ~~~ \teecccs(\lambda, \comp, N) =\frac{\trcccs(\lambda, \comp, N)}{\tecccs(\lambda, \comp, N)} \nonumber \\
& \hspace{-0.7cm}  s.t. ~~~~ \pfcccs \leq \overline{P}_{f}, \nonumber \\
& \hspace{-0.7cm} ~~~~~~~~ \pdcccs \geq \overline{P}_{d}. \label{OP2withc}
\end{eqnarray} 
Later, in Sec.~\ref{SecResults}, we show that the errors due to these approximations are negligible.

Note that $\pdcccs$ and $\pfcccs$ are dependent on $\comp$ and $N$, only through their product $\comp N$. The following theorem provides the solution to the optimal threshold, $\lambda^*$, for the optimization problems in \eqref{OP2withN} and \eqref{OP2withc}.

\begin{theorem} 
\label{OptLambdaRan}
The optimal threshold $\lambda^*$ for the optimization problem $\mathcal{OP}1^{(c)}_{\scriptscriptstyle CCCS}$ satisfies the constraint $\pdcccs \geq \overline{P}_{d}$ with equality, and is given by
\begin{align}
& \lambda^* = \varn (1+\gamma) \left\{\sqrt{2 \comp N P} Q^{-1}(\overline{P}_d) + \comp N P\right\}.
\end{align}
\end{theorem}
\begin{proof}
See Appendix \ref{OptLambdaRanNPThmProof}.
\end{proof}

As a consequence of the above theorem, we now show that the other constraint in  \eqref{OP2withc}, namely $\pfcccs \leq \overline{P}_f$, reduces to an upper bound on the product $\comp N$. By substituting $\lambda = \lambda^*$ in the constraint $\pfcccs \leq \overline{P}_f$, we get
\begin{align}
& \overline{P}_f \geq Q\left(\frac{\frac{\varn (1+\gamma) \left\{\sqrt{2 \comp N P} Q^{-1}(\overline{P}_d) + \comp N P\right\}}{\varn}-\comp N P}{\sqrt{2 \comp N P}}\right).
\end{align}
Rearranging the above equation, this condition reduces to
\begin{align}
\comp N \leq \frac{2}{\gamma^2 P}\left\{Q^{-1}(\overline{P}_f) - (1+\gamma) Q^{-1}(\overline{P}_d)\right\}^2.
\end{align}

Now, the optimization problem $\mathcal{OP}1^{(\comp)}_{\scriptscriptstyle CCCS}$ given in  \eqref{OP2withc} can be reformulated as
\begin{align}
& \mathcal{OP}2^{(\comp)}_{\scriptscriptstyle CCCS}:  \underset{\comp}{\max} ~~~ \teecccs(\lambda^*, \comp, N) \nonumber \\
& ~~ s.t. ~ \comp \leq \compbnd \triangleq \frac{2 \left\{Q^{-1}(\overline{P}_f) - (1+\gamma) Q^{-1}(\overline{P}_d)\right\}^2}{\gamma^2 N P}.  \label{OP3withc}
\end{align}

In the next theorem, we consider  \eqref{OP3withc} in particular, and show that the corresponding objective function is monotonically increasing (and concave) for $\comp \in (0,\compbnd)$, for a given $N$. Therefore, the optimal $\comp^*$ which maximizes $\teecccs(\lambda^*, \comp, N)$ for a given $N$ is given as $\comp^* = \compbnd$.

\begin{theorem} \label{cthmranNP}
For a given $N$, the objective function in the optimization problem $\mathcal{OP}2^{(\comp)}_{\scriptscriptstyle CCCS}$ is monotonically increasing in $\comp \in (0,\compbnd)$. Therefore, $\comp^* = \compbnd$.
\end{theorem}
\begin{proof}
See Appendix \ref{cthmranNPProof}.
\end{proof}

A similar argument can be made for the problem in \eqref{OP2withN}, using the following theorem.
	
\begin{theorem} \label{NthmranNP}
For a given $\comp$, the objective function in the optimization problem $\mathcal{OP}2^{(N)}_{\scriptscriptstyle CCCS}$ is monotonically increasing in $N \in (0,\nbnd)$. Therefore, $N^* = \nbnd$.
\end{theorem}
\begin{proof}
The proof is in similar lines to that of Theorem \ref{cthmranNP}, and is omitted for brevity.
\end{proof}

To find the optimal operating point -- either $(N^{*},\comp)$, or $(N,\comp^{*})$ -- that maximizes the energy efficiency based on the above analytic development, we propose the following simple search algorithm. Summarized as Algorithm~\ref{OptCCCSAlgooptimum}, this technique can be used to solve either of the optimization problems $\mathcal{OP}2^{(\comp)}_{\scriptscriptstyle CCCS}$ or $\mathcal{OP}2^{(N)}_{\scriptscriptstyle CCCS}$. This completes our analysis on finding the optimal $\comp^*$ for a given $N$, or to find the optimal $N^*$ for a given $\comp$, such that the energy efficiency of the CRN is maximized. In the next section, we consider a similar performance analysis of the CRN with a deterministic PU signal.

\section{Performance with Deterministic PU Signal} \label{EEforDetersignal}

In this section, we consider the EE performance of the CR network for the case when PU signal is deterministic. Although unrealistic in practice, performance study of a CRN with a deterministic PU signal has been studied earlier in the context of capacity analysis \cite{Urkowitz_IEEE_1967}, spectrum sensing \cite{Reisi_IEEE_2012}, etc., which serves as an upper bound on the performance of a system employed in practice. In the case of a deterministic PU signal, asymptotic distribution of the test statistic at the FC under either hypotheses can be written as \cite{varshney_IEEE_2017}
\begin{eqnarray} \label{teststati}
T(\xmat) \triangleq \sim \left\lbrace \begin{array}{cl}\mathcal {N}(0, \varn N\Vert \phat \xvec\Vert _2^2), & \hbox{under}\ \mathcal{H}_{0} \\ \mathcal {N}(N\Vert \phat \xvec \Vert _2^2, \varn N\Vert \phat \xvec \Vert _2^2) & \hbox{under}\ \mathcal{H}_{1} \end{array}\right.,
\end{eqnarray}
where  $\Vert \phat \xvec \Vert _2^2 \triangleq \xvec^T \meas ^T(\meas \meas^T)^{-1}\meas \xvec$. From \eqref{teststati}, the probabilities of false-alarm and signal detection at the FC following the CCCS scheme with deterministic PU signal are given by
\begin{align}
& \pfcccsdet = P(T(\xmat) > \lambda |\mathcal{H}_{0}) \nonumber \\
& \phantom{\pfcccsdet} = Q\left(\frac{\lambda-N \comp \gamma \varn}{\sqrt{\varn N\Vert \phat \xvec \Vert _2^2)}}\right) \\
& \pdcccsdet=P(T(\xmat)>\lambda|\mathcal{H}_{1}) \nonumber \\
& \phantom{\pdcccsdet} = Q\left(\frac{\lambda-N \comp \gamma \varn}{\sqrt{\varn N\Vert \phat \xvec \Vert _2^2)}}\right)
\end{align}
As discussed in the random signal case, using the concept of $\epsilon$-stable embedding, for larger value of $NM$ the approximation $\Vert \phat \xvec \Vert_2^2 \approx \frac{M}{P}\Vert \xvec \Vert _2^2 = \comp \Vert \xvec \Vert _2^2$ \cite{varshney_IEEE_2017}. Therefore,
\begin{align} 
& \pfcccsdet \hspace{-0.1cm} = \hspace{-0.1cm} Q\left( \hspace{-0.1cm} \frac{\lambda}{\varn \sqrt{N \comp \gamma}} \hspace{-0.05cm} \right), \\
& \pdcccsdet \hspace{-0.1cm} = \hspace{-0.1cm} Q\left( \hspace{-0.1cm} \frac{\lambda-N \comp \gamma\varn}{\sigma_{w}^2\sqrt{N \comp \gamma}} \hspace{-0.1cm} \right).
\end{align}
It is easy to show that the detection threshold $\lambda = \frac{N}{2} \xvec^T \meas^T(\meas \meas^T)^{-1} \meas \xvec = \frac{N}{2}\Vert \phat \xvec \Vert _2^2=\frac{N}{2} \comp \gamma \sigma_{n}^2$. Therefore, the final expressions for $\pfcccsdet$ and $\pdcccsdet$ are given by
\begin{align} 
\label{pdpfequdet}
& \pfcccsdet =  Q\left(\frac{\sqrt{\comp N \gamma}}{2}\right), \\
& \pdcccsdet = Q\left(-\frac{\sqrt{\comp N \gamma}}{2}\right) 
\end{align} 
Note that the expressions for average achievable throughput, average energy consumption and the energy efficiency expressions across all four scenarios $\textbf{S1}-\textbf{S4}$ for the deterministic case remains similar to the random case, except that $\pfcccs$ and $\pdcccs$ are replaced by $\pfcccsdet$ and $\pdcccsdet$, respectively. The approximations discussed in the previous case also hold for the deterministic case. For a given $\comp$, the corresponding optimization problem for the deterministic case can be written as
\begin{eqnarray}               
& \mathcal{OP}1^{(N)}_{\scriptscriptstyle CCCS, det}:  \underset{N}{\max} ~~~ \teecccsdet(\comp, N) =\frac{\trcccsdet(\comp, N)}{\tecccsdet(\comp,N)} \nonumber \\
& \hspace{-0.4cm}  s.t. ~~ \pfcccsdet \leq \overline{P_{f}}, \nonumber \\
& \hspace{-0.4cm} ~~~~~~ \pdcccsdet \geq \overline{P_{d}}, \label{OP2withN_det}
\end{eqnarray}
and the optimization problem for given $N$ is given by
\begin{eqnarray}               
& \mathcal{OP}1^{(\comp)}_{\scriptscriptstyle CCCS, det}:  \underset{\comp}{\max} ~~~ \teecccsdet(\comp, N) =\frac{\trcccsdet(\comp, N)}{\tecccsdet(\comp,N)} \nonumber \\
& \hspace{-0.4cm}  s.t. ~~ \pfcccsdet \leq \overline{P_{f}}, \nonumber \\
& \hspace{-0.4cm} ~~~~~~ \pdcccsdet \geq \overline{P_{d}}, \label{OP2withc_det}
\end{eqnarray} 
for some $0 < \pfcccsdet < \pdcccsdet < 1$. We later show that the errors due to these approximations are negligible. Again, note that both $\pfcccsdet$ and $\pdcccsdet$ depend on $\comp$ and $N$ through the product $\comp N$. 

\begin{theorem} \label{OptLambdaDetNPThm}
The optimal threshold $\lambda^*$ for the optimization problems $ \mathcal{OP}1^{(N)}_{\scriptscriptstyle CCCS, det}$ and $ \mathcal{OP}1^{(\comp)}_{\scriptscriptstyle CCCS, det}$ satisfies the constraint $\pdcccs \geq \overline{P}_{d}$ with equality, and is given by
\begin{align}
& \lambda^* = \varn \sqrt{N\comp \gamma} \left\{Q^{-1}(\overline{P}_d) + \sqrt{N\comp \gamma}\right\}.
\end{align}
\end{theorem}
\begin{proof}
See Appendix \ref{OptLambdaDetNPThmProof}.
\end{proof}

\begin{algorithm} [t!]
	\caption{Algorithm for optimizing N and $\comp$} \label{OptCCCSAlgooptimum}
	\begin{algorithmic}[1]
		\State Set $P_{s}, P_{t}, \ttot, \tau_s,\tau_s, N, i_0,  \phn, \pho,  \overline{P_f}, \overline{P_d}$
		\State When $\pfcccs \leq\overline{P}_f$, Calculate $\comp N$ using (\ref{pfeqn})
		\State  When $\pdcccs \geq \overline{P}_d$, Calculate $\comp N$ using (\ref{pdeqn})
		\Procedure   \text{To find optimal $\comp$}
		\State Fix $N \in (1,N_{\max})$
		\State Compute ~~~~~~$\comp=\frac{2 \left\{Q^{-1}(\overline{P}_f) - (1+\gamma) Q^{-1}(\overline{P}_d)\right\}^2}{\gamma^2 N P}$ ~~~~~~~~~~~~~~~(\text {for random signal})
		\State Compute ~~~~~~$\comp= \frac{ \left\{Q^{-1}(\overline{P}_f) - Q^{-1}(\overline{P}_d)\right\}^2}{\gamma N }$ ~~~~~~~~~~~~~~~(\text {for deterministic signal})
		\State Compute $\teecccs(\comp, N)$ using \eqref{EECCCSapproxEqn}
		with $N$ and $\comp$
		\State Compute $\max(\teecccs(\comp, N))$ and respective $\comp^{*}$
		\State \textbf{return} ~ $\comp^{*}$
		
		\EndProcedure
		\Procedure   \text{To find optimal $N$}
		\State Fix $\comp \in (0,1)$
		\State Compute ~~~$N=\frac{2 \left\{Q^{-1}(\overline{P}_f) - (1+\gamma) Q^{-1}(\overline{P}_d)\right\}^2}{\gamma^2 \comp P}$ ~~~~~~~~~~~~~~~(\text {for random signal})
			\State Compute $N= \frac{ \left\{Q^{-1}(\overline{P}_f) - Q^{-1}(\overline{P}_d)\right\}^2}{\gamma \comp }$ ~~~~~~~~~~~~~~~~~~(\text {for deterministic signal})
		\State Compute $\teecccs(\comp, N)$ using \eqref{EECCCSapproxEqn}
		with $N$ and $\comp$
		\State Compute $\max(\teecccs(\comp, N))$ and respective $N^{*}$
		\State \textbf{return} ~ $N^{*}$
		\EndProcedure
		\State Return $\max(\eecccs)$ and the corresponding $N^{*},\comp^{*}$,
	\end{algorithmic}
\end{algorithm}

Similar to the case of random PU signal, following the above theorem, we now show that the other constraint in \eqref{OP2withN_det} and \eqref{OP2withc_det}, namely $\pfcccs \leq \overline{P}_f$, reduces to an upper bound on the product $\comp N$. By substituting $\lambda = \lambda^*$ in the constraint $\pfcccs \leq \overline{P}_f$, we get
\begin{align}
& \overline{P}_f \geq Q\left(\frac{\sqrt{ \comp N \gamma}\varn(Q^{-1}(\overline{P}_d)+\sqrt{N\comp\gamma})}{\varn\sqrt{ \comp N \gamma}}\right) 
\end{align}
Rearranging the above equation, this condition reduces to
\begin{align}
\comp N \leq \frac{1}{\gamma }\left\{Q^{-1}(\overline{P}_f) -  Q^{-1}(\overline{P}_d)\right\}^2.
\end{align}

Now, the optimization problems $\mathcal{OP}1^{(N)}_{\scriptscriptstyle CCCS, det}$ and $\mathcal{OP}1^{(\comp)}_{\scriptscriptstyle CCCS, det}$ given in \eqref{OP2withN_det} and \eqref{OP2withc_det} can be respectively reformulated as
\begin{align}
& \mathcal{OP}2^{(N)}_{\scriptscriptstyle CCCS,det}:  \underset{N}{\max} ~~~ \teecccs(\lambda^*, \comp, N) \nonumber \\
& ~~~~~  s.t. ~ N \leq \nbnd \triangleq \frac{ \left\{Q^{-1}(\overline{P}_f) - Q^{-1}(\overline{P}_d)\right\}^2}{\gamma \comp },  \label{OP3withNdet}
\end{align}
and
\begin{align}
& \mathcal{OP}2^{(\comp)}_{\scriptscriptstyle CCCS,det}:  \underset{\comp}{\max} ~~~ \teecccs(\lambda^*, \comp, N) \nonumber \\
& ~~~~~  s.t. ~ \comp \leq \compbnd \triangleq \frac{ \left\{Q^{-1}(\overline{P}_f) - Q^{-1}(\overline{P}_d)\right\}^2}{\gamma N },  \label{OP3withcdet}
\end{align}

In the next theorem, we consider the problem \eqref{OP3withcdet} in particular, and show that the corresponding objective function is monotonically increasing (and concave) for $\comp \in (0,\compbnd)$, for a given $N$.

\begin{theorem} \label{etathm}
For a given $N$, the objective function in the optimization problem $\mathcal{OP}1^{(\comp)}_{\scriptscriptstyle CCCS, det}$ is monotonically increasing in $\comp \in (0,\compbnd)$, and hence $\comp^* = \compbnd$.
\end{theorem}
\begin{proof}
See Appendix \ref{etathmProof}.
\end{proof} 

A similar argument can be made for the problem in \eqref{OP3withNdet}.

\begin{theorem} \label{etathmdet}
For a given $\comp$, the objective function in the optimization problem $\mathcal{OP}1^{(N)}_{\scriptscriptstyle CCCS, det}$ is monotonically increasing in $N \in (0,\nbnd)$, and hence $N^* = \nbnd$.
\end{theorem}
\begin{proof}
The proof is in similar lines to that of Theorem \ref{cthmranNP}, and is omitted for brevity.
\end{proof}
\section{Numerical Results and Discussion}  \label{SecResults}

\begin{figure*}[t]
\centering
\begin{subfigure}[b]{0.5\textwidth}
\centering
        \includegraphics[height=3in]{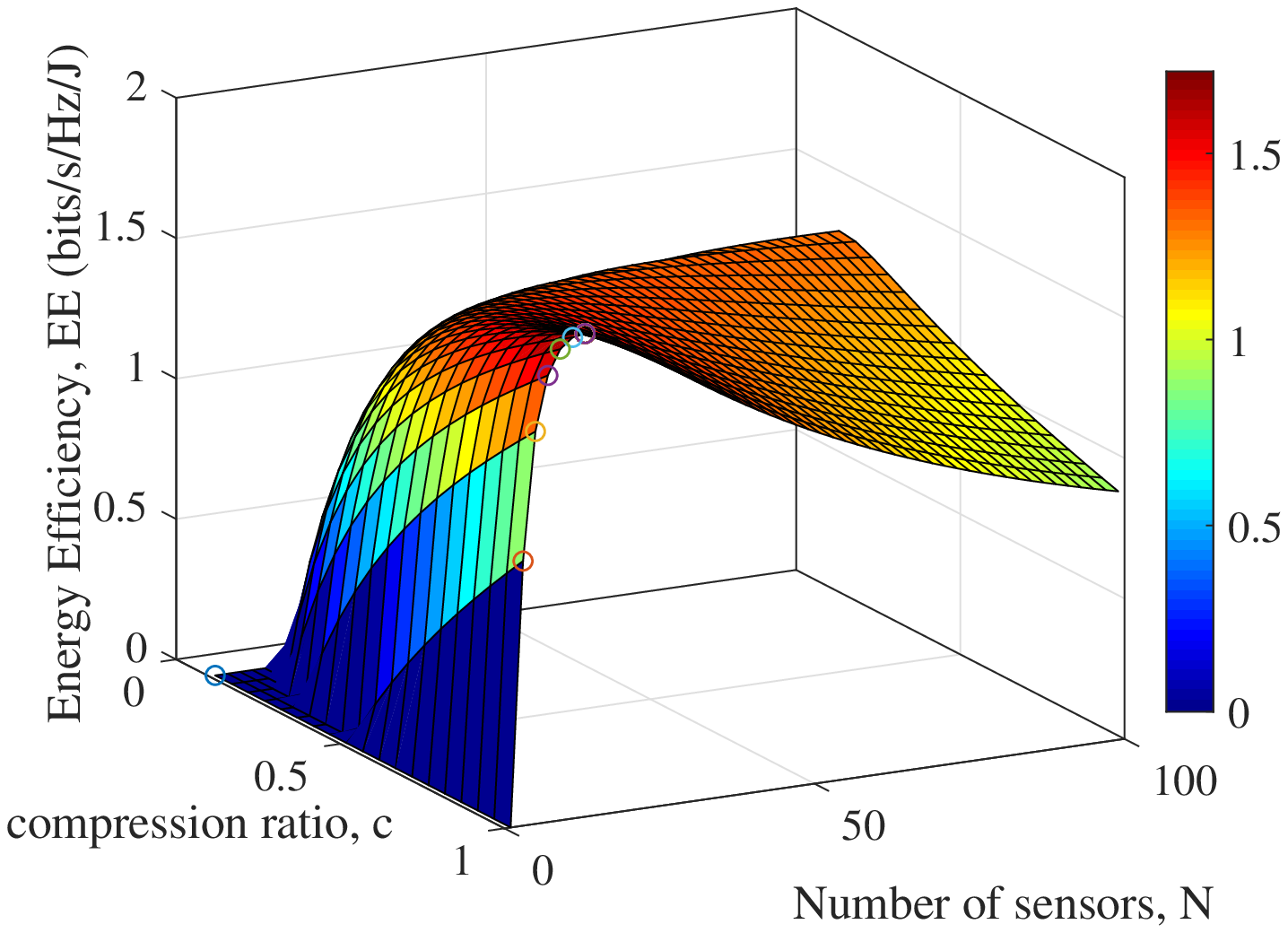}
        \caption{Deterministic PU signal.}
    \end{subfigure}%
    ~ 
    \begin{subfigure}[b]{0.5\textwidth}
        \centering
        \includegraphics[height=3in]{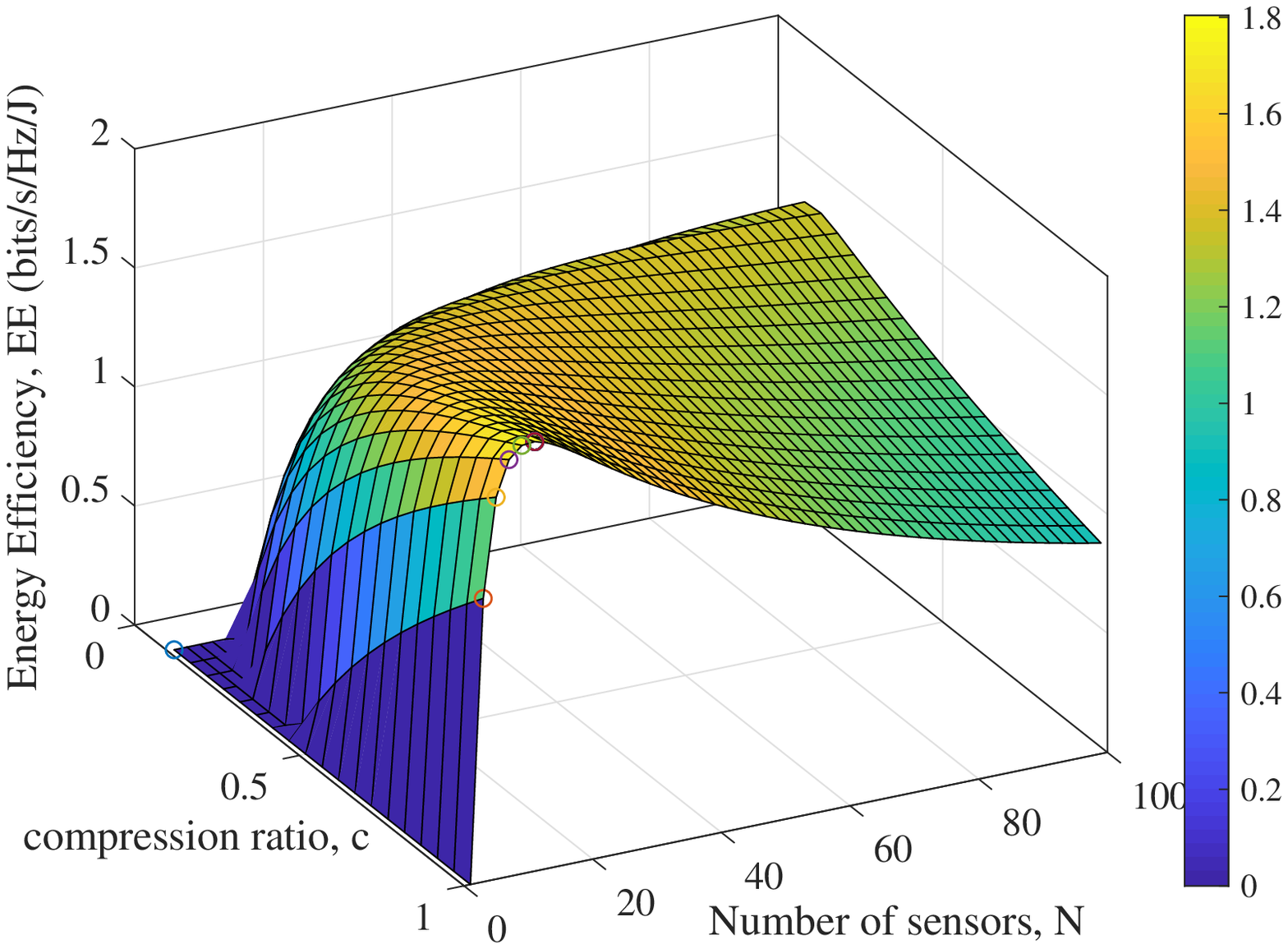}
        \caption{Random PU signal.}
    \end{subfigure}
    \caption{Energy efficiency as a function of number of sensors $N$ and compression ratio $\comp$ for (a) deterministic signal case, SNR = $-3$ dB (b) random signal case SNR = $-9$ dB.}
    \label{3Dplot}
\end{figure*}

\begin{figure*}[t!]
    \centering
    \begin{subfigure}[b]{0.5\textwidth}
        \centering
        \includegraphics[height=2.56in]{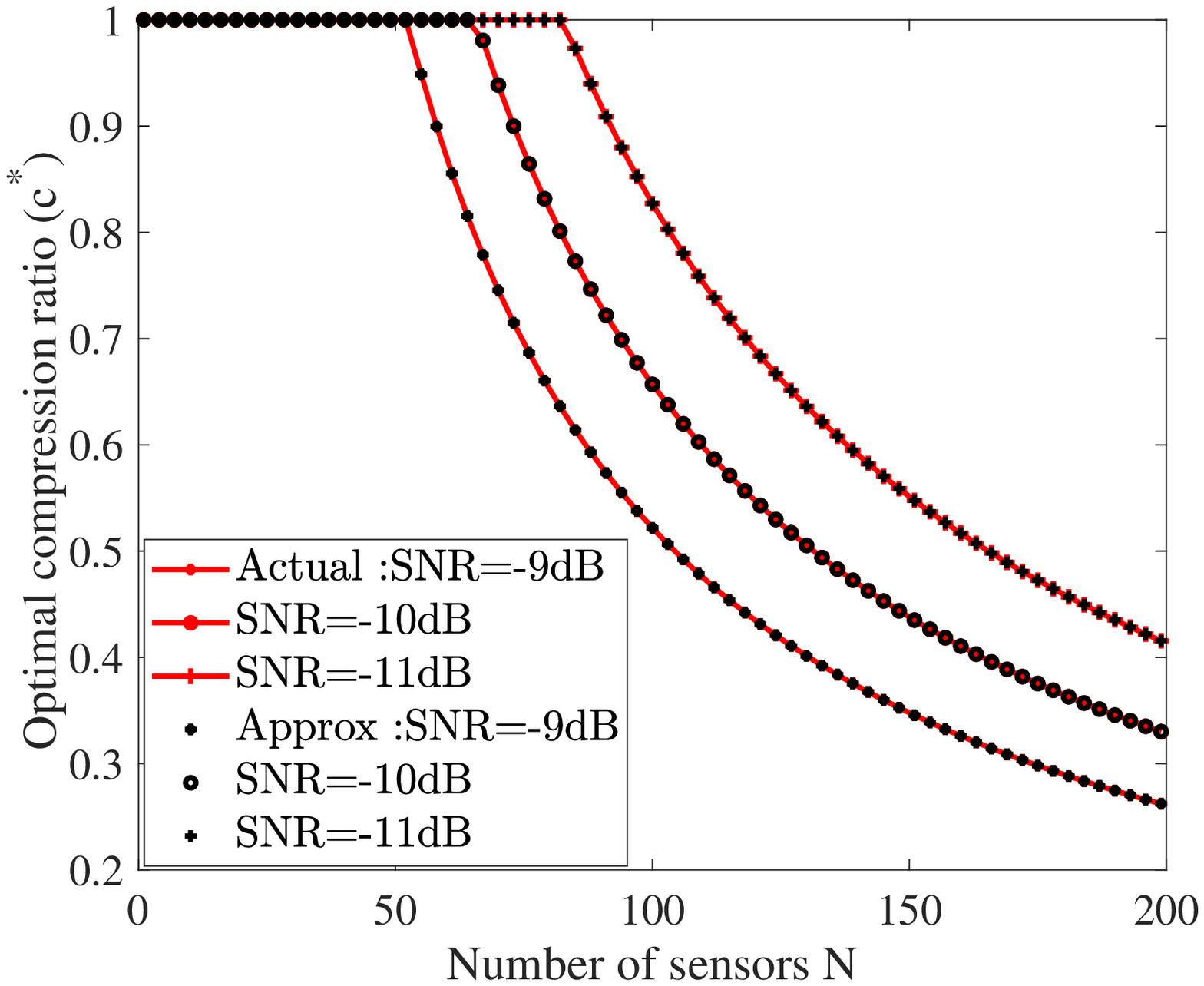}
        \caption{Deterministic PU signal.}
         \label{FigcvsNa}
    \end{subfigure}%
    ~ 
    \begin{subfigure}[b]{0.5\textwidth}
        \centering
        \includegraphics[height=2.56in]{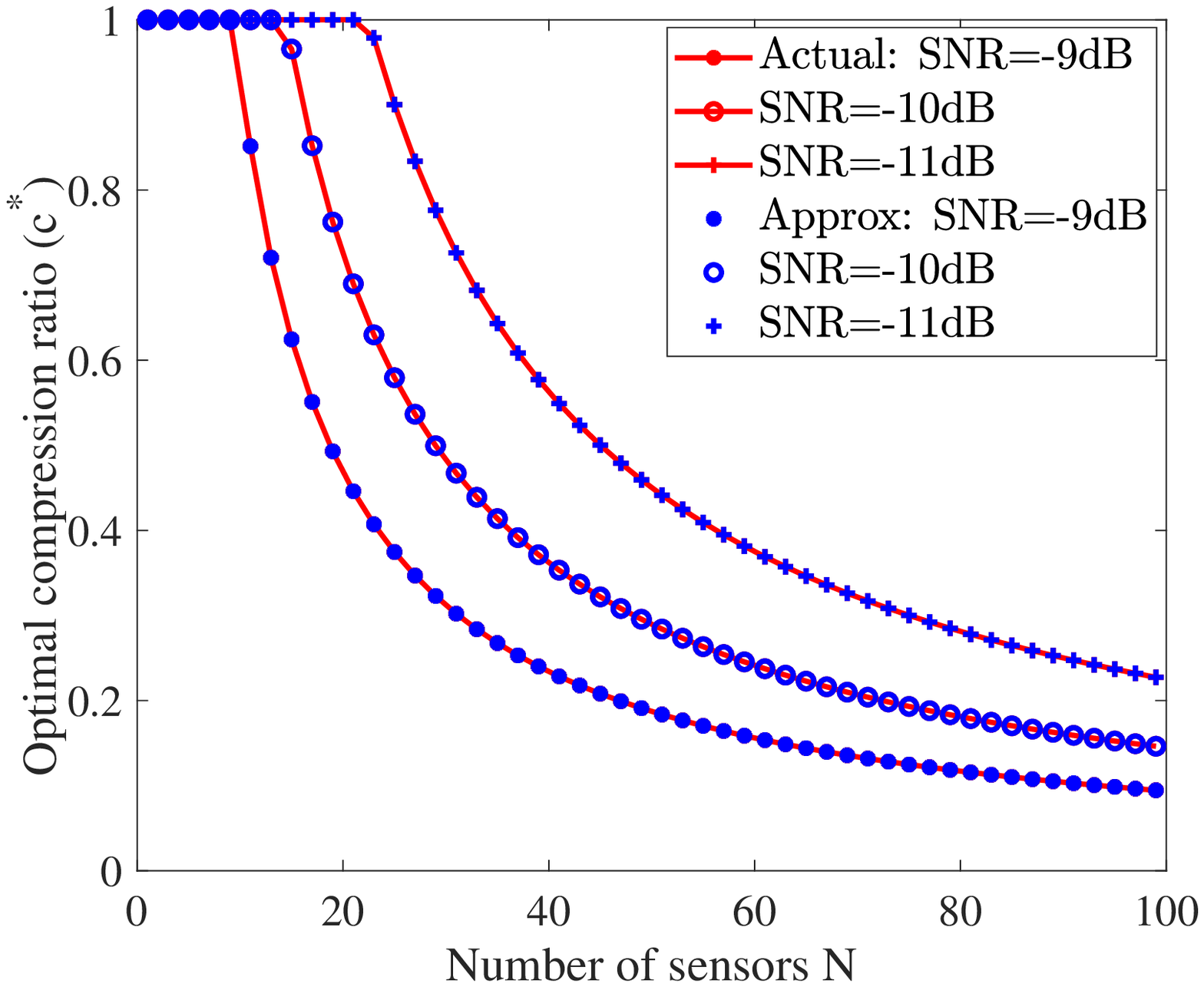}
        \caption{Random signal.}
        \label{FigcvsNb}
    \end{subfigure}
    \caption{Variation of the optimal compression ratio ($\comp^*$) with number of sensors $N$ for (a) deterministic signal case (b) random signal case.}
\end{figure*}

\begin{figure*}[t!]
	\centering
	\begin{subfigure}[b]{0.5\textwidth}
		\centering
		\includegraphics[height=2.57in]{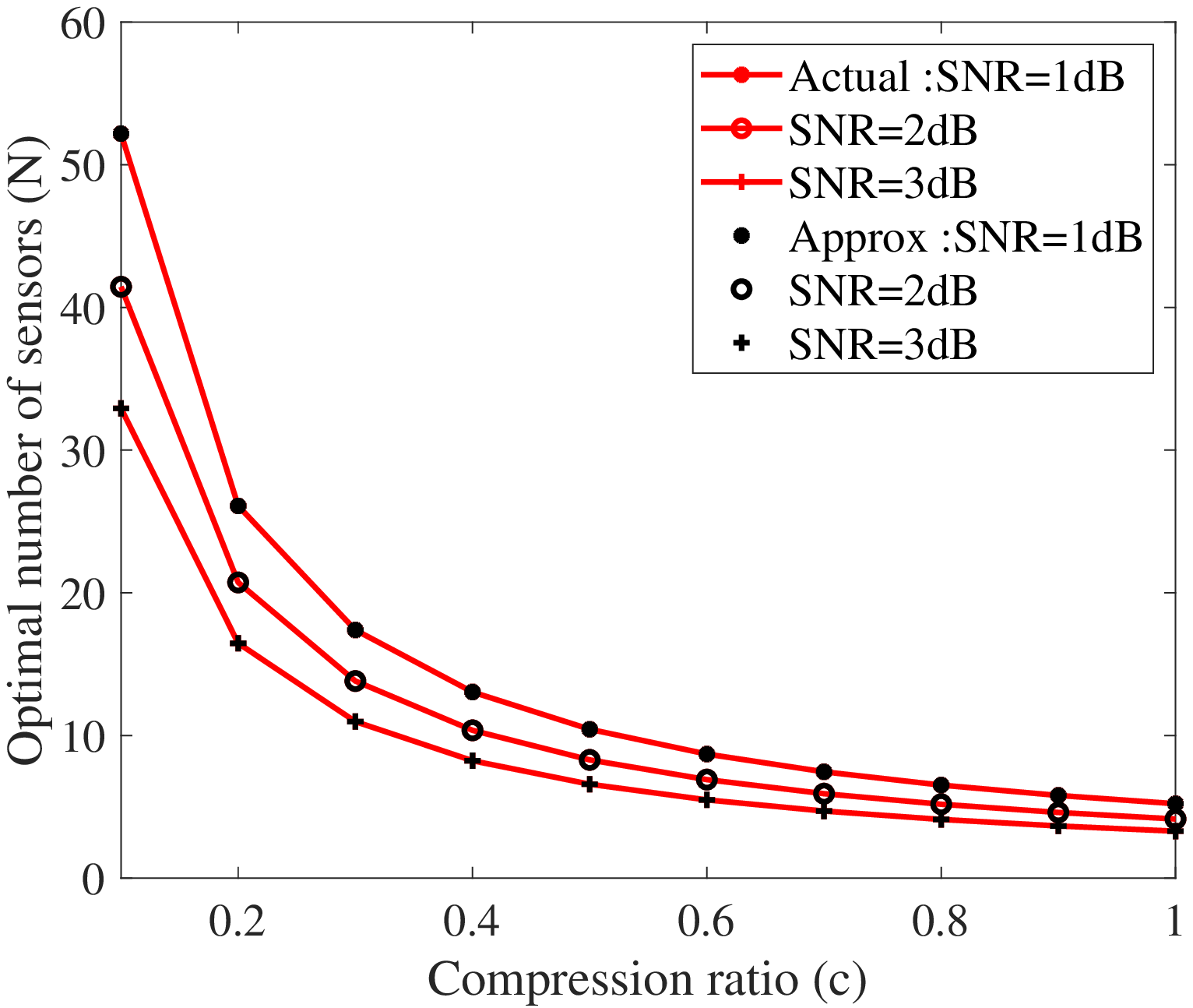}
		\caption{Deterministic PU signal.}
		\label{FigNvsca}
	\end{subfigure}%
	~ 
	\begin{subfigure}[b]{0.5\textwidth}
		\centering
		\includegraphics[height=2.56in]{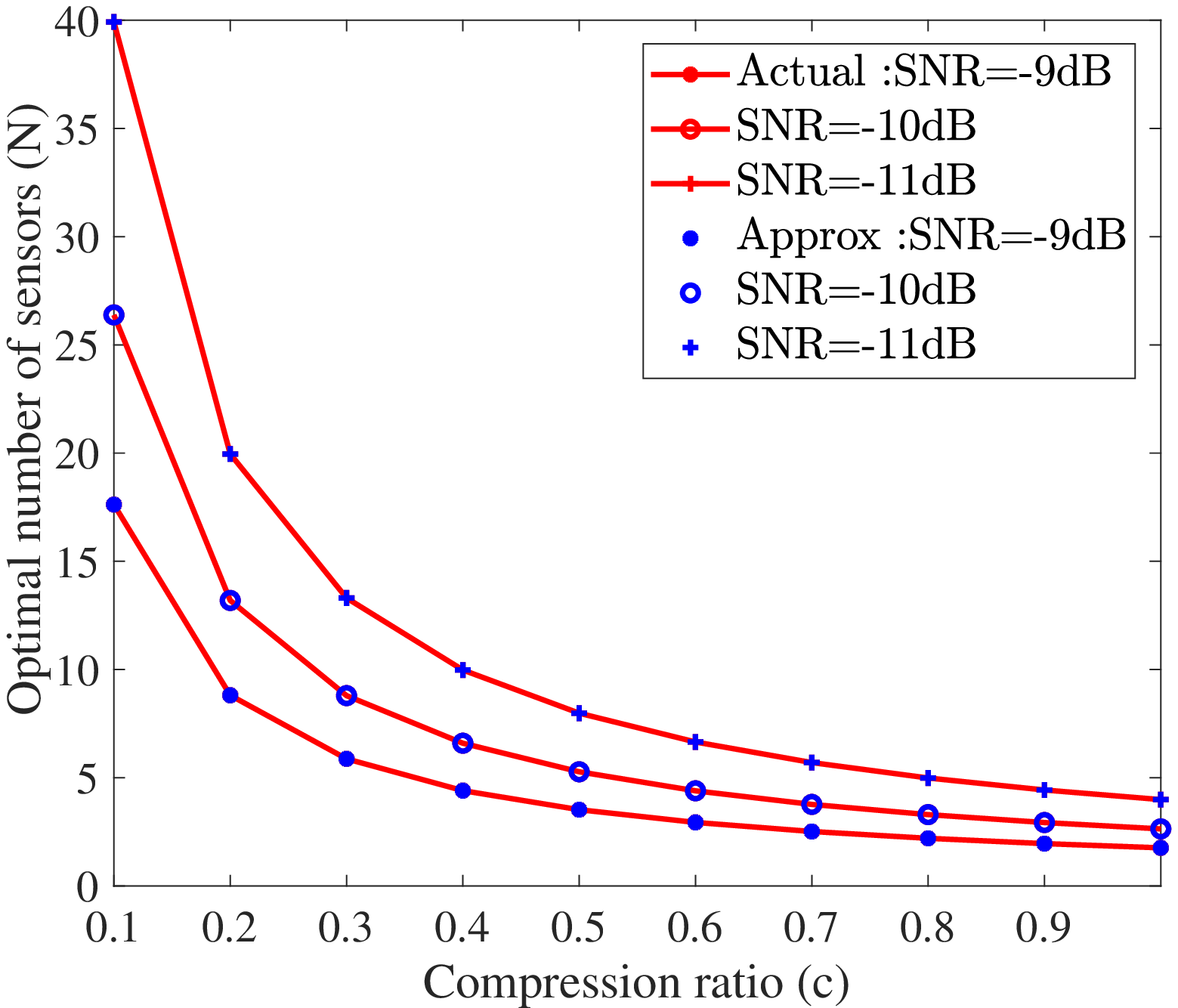}
		\caption{Random PU signal.}
		\label{FigNvscb}
	\end{subfigure}
	\caption{Variation of the optimal number of sensors ($N^*$) with compression ratio $\comp^*$ for (a) deterministic signal case (b) random signal case.}
\end{figure*}

\begin{figure*}[t!]
    \centering
    \begin{subfigure}[b]{0.5\textwidth}
        \centering
        \includegraphics[height=2.56in]{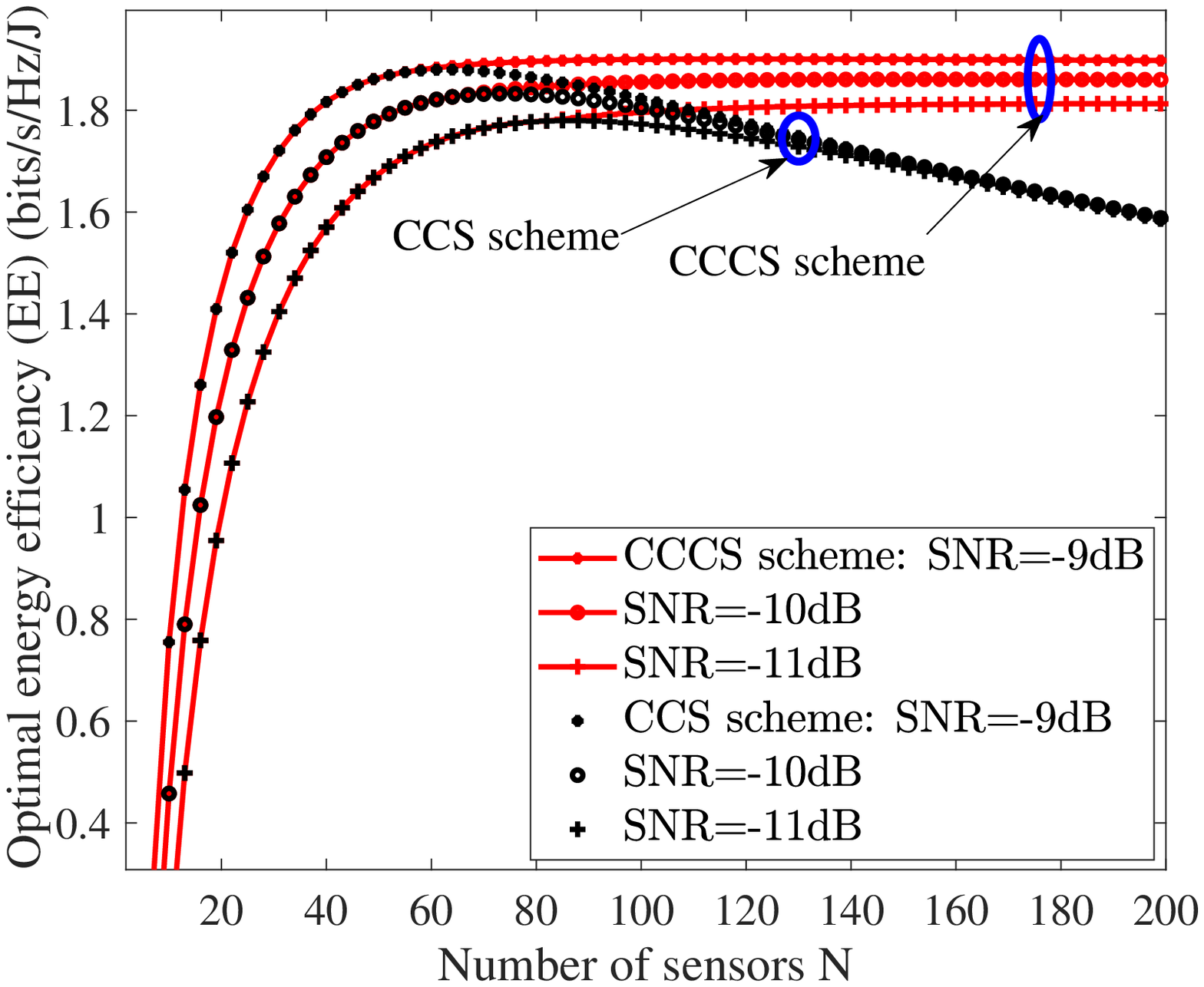}
        \caption{Deterministic PU signal.}
        \label{EEplotdet}
    \end{subfigure}%
    ~ 
    \begin{subfigure}[b]{0.5\textwidth}
        \centering
        \includegraphics[height=2.56in]{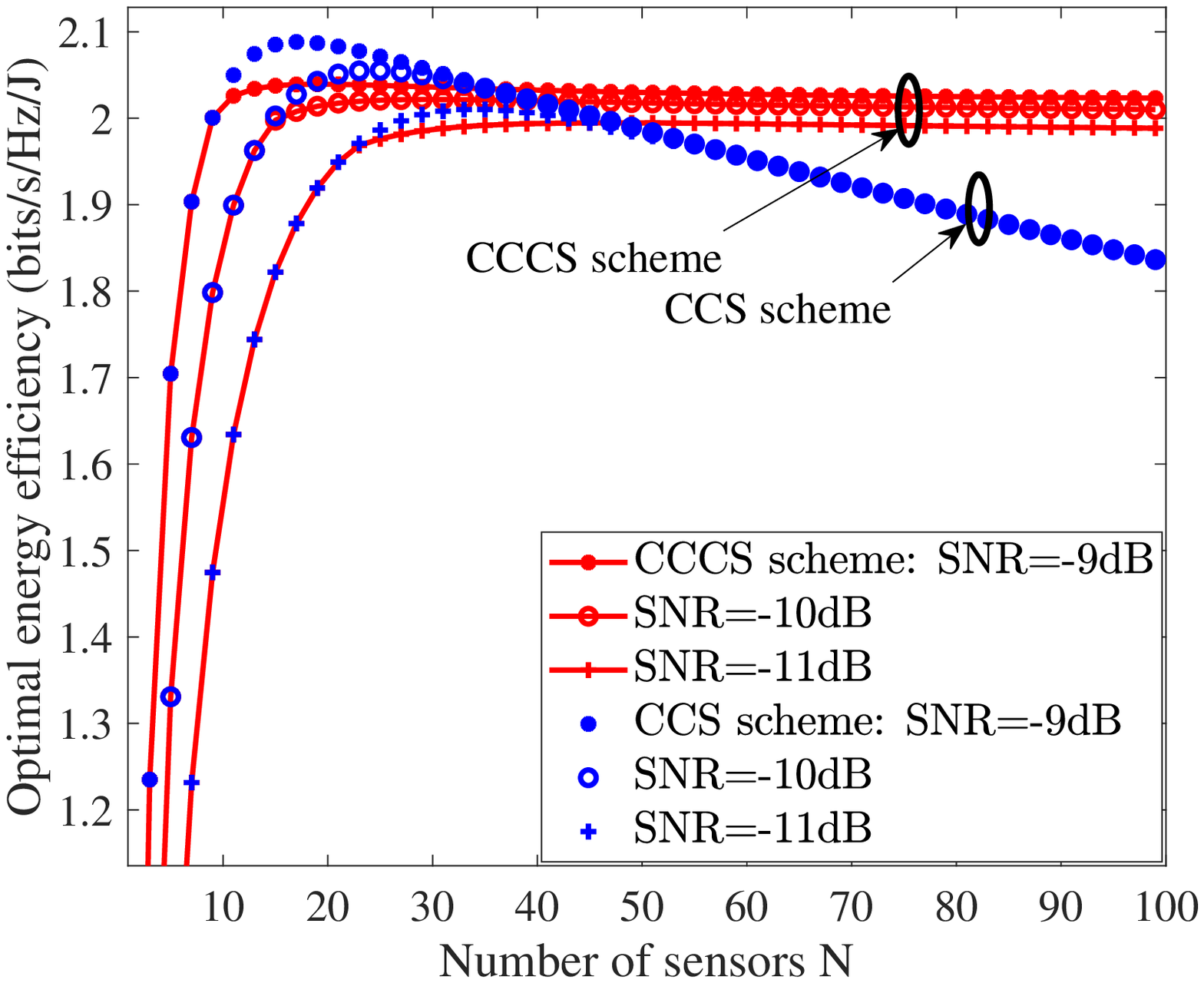}
        \caption{Random PU signal.}
        \label{EEplotran}
    \end{subfigure}
    \caption{Variation of optimal energy efficiency with number of sensors $N$ for (a) deterministic signal case (b) random signal case.}
    
\end{figure*}

\begin{figure*}[t!]
    \centering
    \begin{subfigure}[b]{0.5\textwidth}
        \centering
        \includegraphics[height=2.56in]{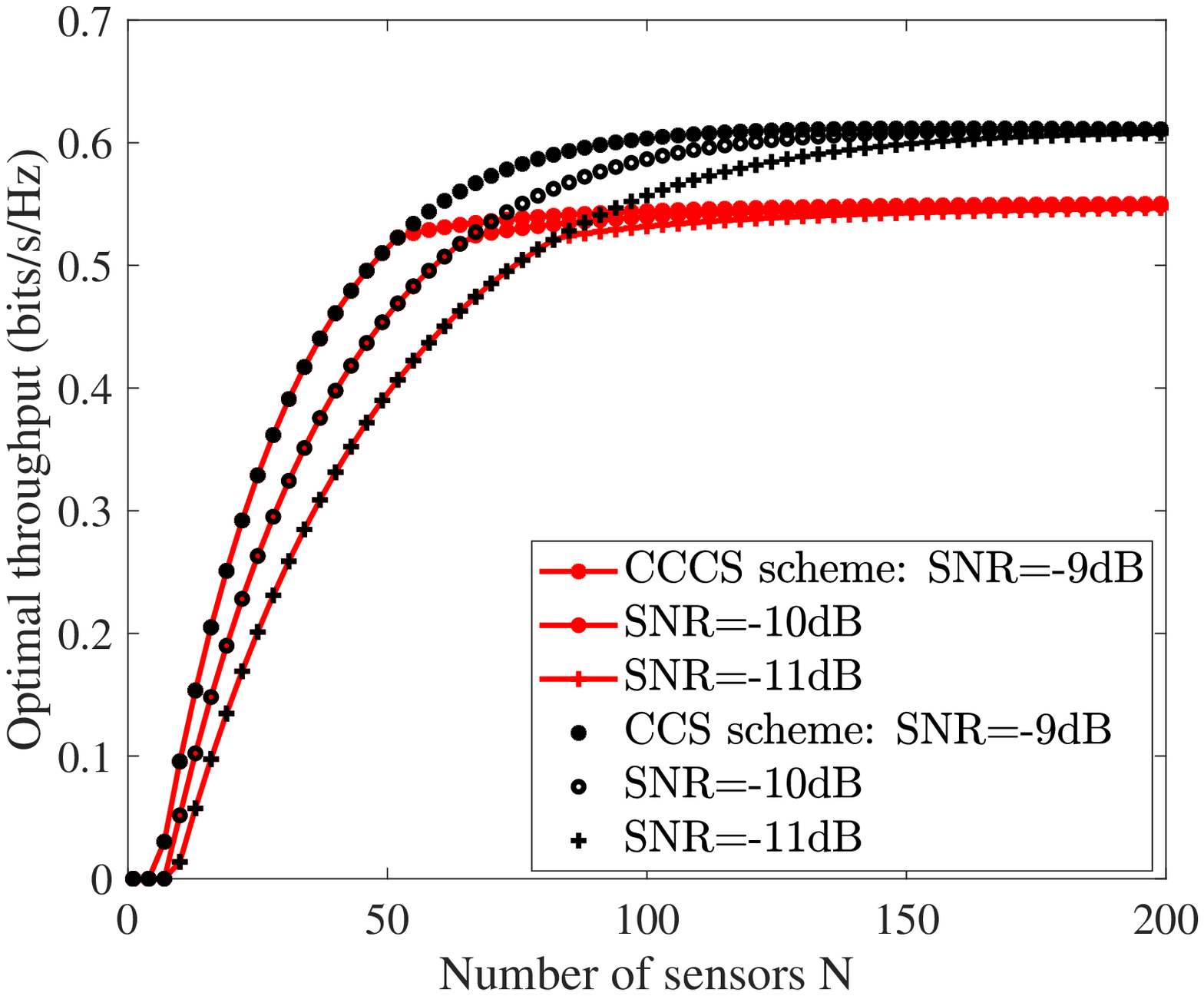}
        \caption{Deterministic PU signal.}
        \label{Rdet}
    \end{subfigure}%
    ~ 
    \begin{subfigure}[b]{0.5\textwidth}
        \centering
        \includegraphics[height=2.56in]{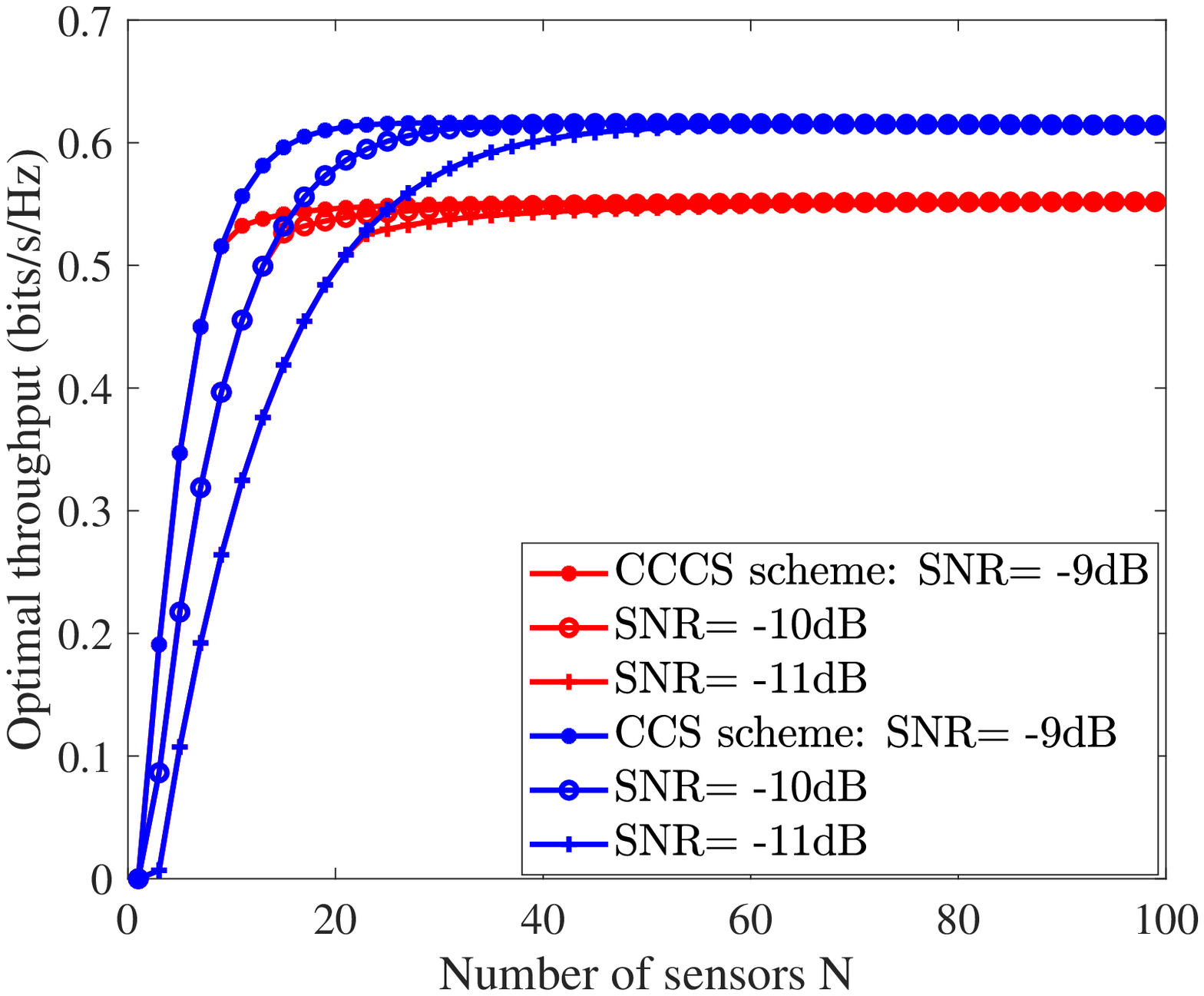}
        \caption{Random PU signal.}
    \end{subfigure}
     \label{Rran}
    \caption{Variation of optimal achievable throughput with number of sensors $N$ for (a) deterministic signal case (b) random signal case.}
     \label{Rcomp&uncomp}
\end{figure*}

\begin{figure*}[t!]
    \centering
    \begin{subfigure}[b]{0.5\textwidth}
        \centering
        \includegraphics[height=2.56in]{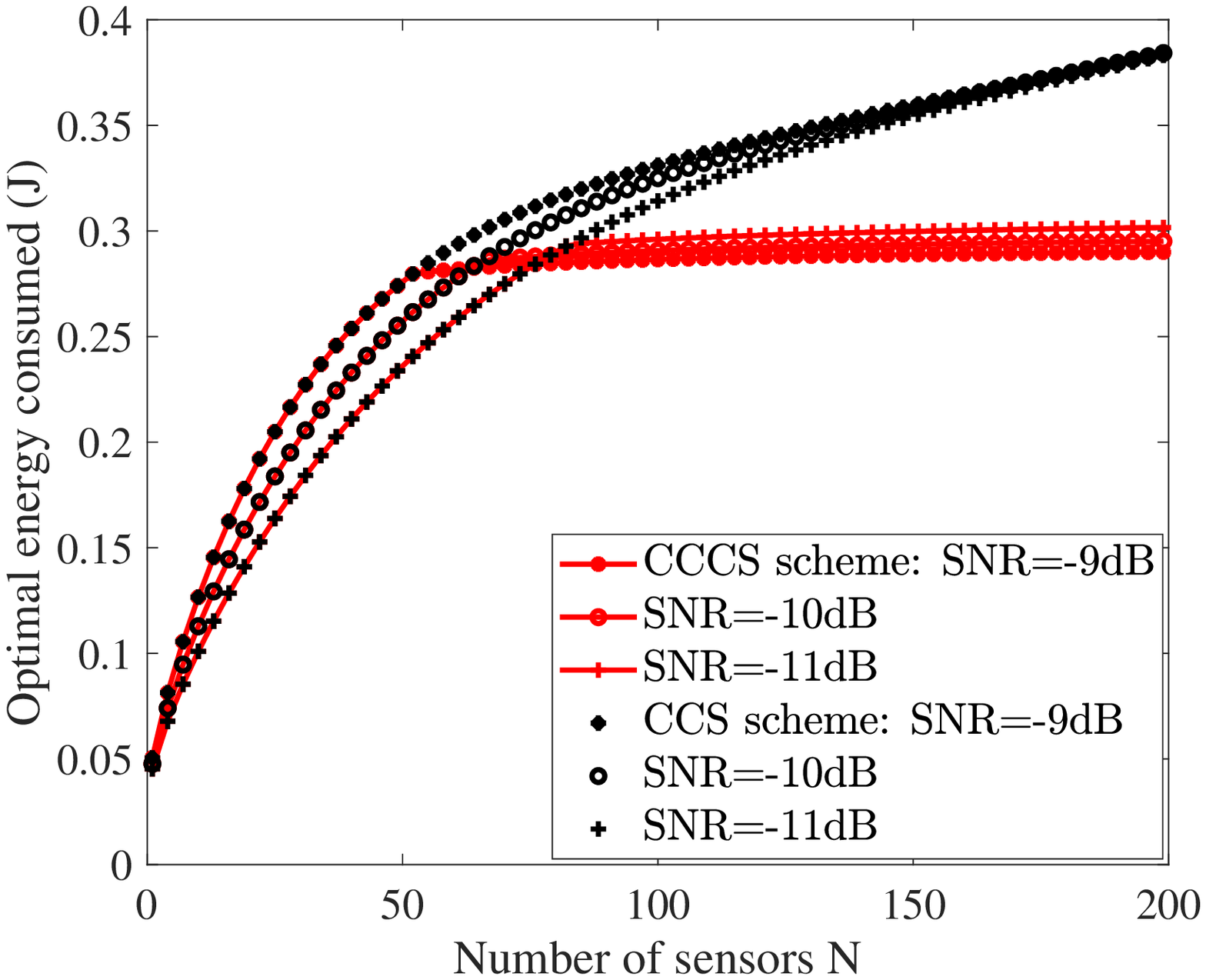}
        \caption{Deterministic PU signal.}
    \end{subfigure}%
    ~ 
    \begin{subfigure}[b]{0.5\textwidth}
        \centering
        \includegraphics[height=2.56in]{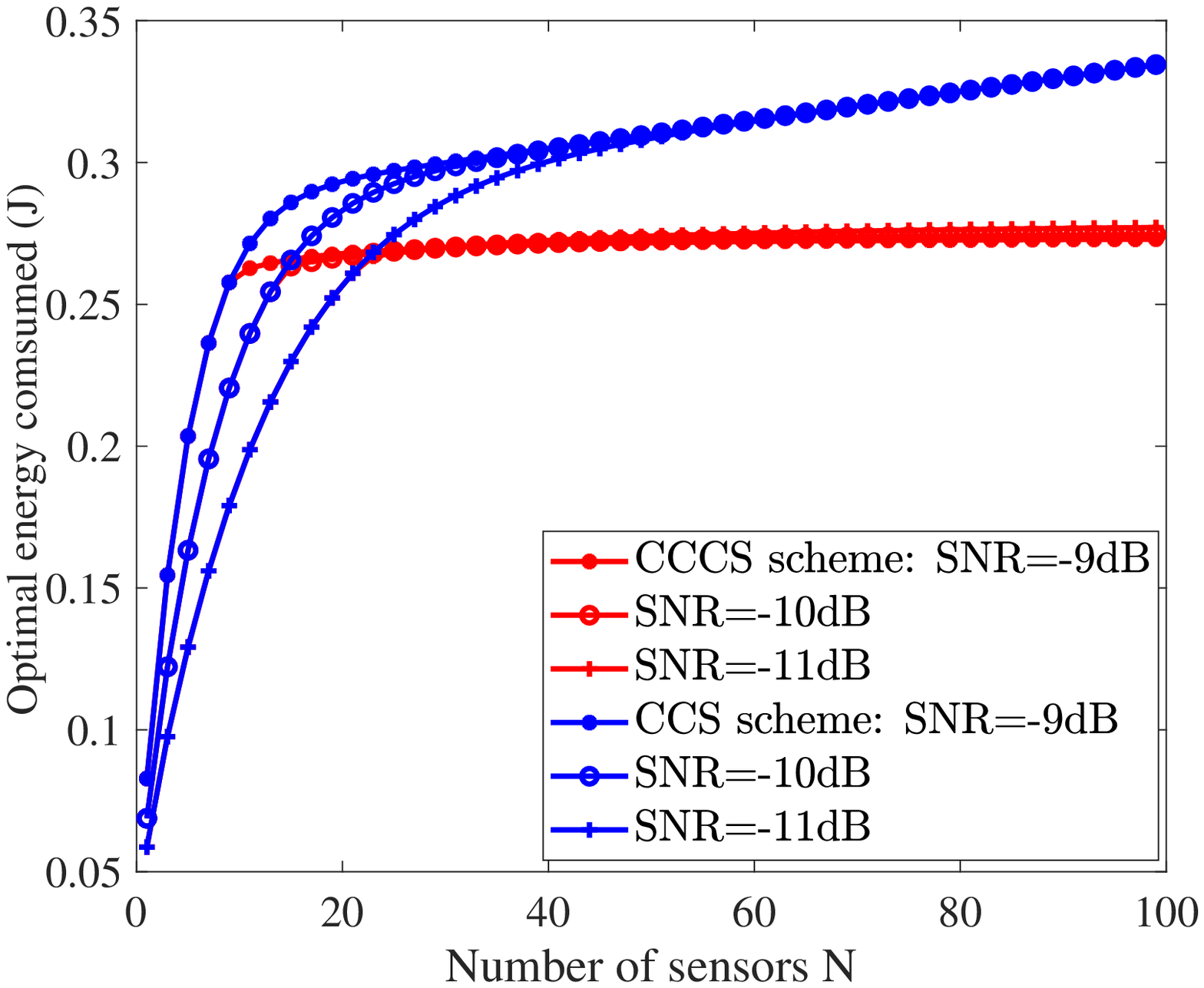}
        \caption{Random PU signal.}
    \end{subfigure}
    \caption{Variation of optimal energy consumption with number of sensors $N$ for (a) deterministic signal case (b) random signal case.}
    \label{energy}
\end{figure*}

In this section, we study the performance of the CCCS technique in comparison with the CCS technique, in terms of energy efficiency and validate our analysis, through numerical techniques. The parameter values are fixed as follows. The target probability of detection, $\overline{P}_{d}$, and false-alarm probability, $\overline{P}_{f}$, are fixed to be $0.9$ and $0.1$, respectively. The prior probabilities $\phn$ and $\pho$ are set to be $0.5$ each. The total frame duration is assumed to be $\ttot=200$ ms \cite{Kishore_Adhoc_2017}. The sampling frequency at the local SUs is assumed to be $f_{s}=1$ MHz, and the sensing power $P_{s}=0.1$ W. The length of the uncompressed received signal vector, $P=100$. The sensing time, $\tau_{s}$, and reporting time, $\tau_{r}$, for the CCCS scheme are set to $30$ ms and $100$ $\mu$s, respectively. The achievable rate of secondary transmission is chosen to be $\mathcal{C} = \log_2(1+\text{SNR}_s) = 6.6582$ bits/sec/Hz, where the SNR at the secondary receiver is assumed to be $\text{SNR}_s = 20$ dB. The transmission power of individual sensors, $P_{t}$, is assumed to be $3$ W. Also, we set the partial throughput factor, $\kappa_c$, and the penalty factor, $\phi$, to be $0.5$ each.

Figure \ref{3Dplot} shows the variation of energy efficiencies for the random and deterministic signal cases, as a function of parameters $\comp$ and $N$. Observe that the energy efficiency is concave in both $\comp$ and $N$. Furthermore, it can be seen that as $N$ increases, $\comp$ decreases, which indicates a better compression. Also, the maximum energy efficiency can be improved with $N$.

Figures \ref{FigcvsNa} and \ref{FigcvsNb} show the variation of the optimal compression ratio $\comp^{*}$ for the CCCS scheme, as a function of $N$ for different values of SNR $\gamma$. First, note that the optimal values of $\comp$ are nearly equal for the actual and approximate energy efficiency values, thereby establishing our earlier claim on the tightness of our approximations involved in evaluation of energy efficiency. The decrease in $\comp^{*}$ with an increase in $N$ is intuitive, since the loss due to compression is recovered in CCCS by increasing $N$, which results in a better throughput, and consequently, a better energy efficiency. Similarly, in Fig.~\ref{FigNvsca} and \ref{FigNvscb}, we consider the variation of optimal $N$ for different values of $\comp$, which yields similar trends and observations.

Figures \ref{EEplotdet} and \ref{EEplotran} show the variation of optimal energy efficiency values with the actual and approximate expressions, for different values of $N$. Note that for low values of $N$, performances of both CCS and CCCS schemes are similar, due to the fact that $\comp^* = 1$ for sufficiently low $N$. As $N$ increases, the system achieves a better compression, and therefore, the performance of CCCS scheme becomes better than that of the CCS scheme. Also, the energy efficiency for both CCS and CCCS schemes increase with an increase in SNR. Moreover, the loss due to the energy efficiency approximation is negligible. Therefore, in our subsequent results, we consider only the approximated energy efficiency values. The reason for a better energy efficiency of the CCCS scheme in comparison to the CCS scheme can possibly be either because CCCS achieves a better throughput, or it achieves a lower energy consumption. Between these two cases, since the detection performance of the CCS scheme is better than that of CCCS scheme for a given $N$ (or $\comp$), the achievable throughput of the CCS scheme will always be higher as compared to that of the CCCS scheme. Therefore, the improvement in the energy efficiency of the CCCS scheme must be due to a significant reduction in energy consumption in comparison to the CCS scheme. Figures \ref{Rcomp&uncomp} and \ref{energy} corroborate this argument. In Fig.~\ref{Rcomp&uncomp}, the achievable throughput of CCS and CCCS schemes are compared, where the former is naturally found to be better. For larger values of $N$, the detection probability and hence the throughput of the CCS scheme improves faster. However, as shown in Fig.~\ref{energy}, the energy consumption of the CCS scheme also increases rapidly with $N$, as opposed to the CCCS scheme, where the increase is much slower since $\comp^*$ decreases with $N$. This is true for both random signal and deterministic signal cases. Hence, in scenarios where the energy consumption has a larger priority in a signal detection scenario CCCS scheme could be preferred. However, in the scenario where the sensing accuracy is a main concern, CCS scheme yields a better performance, in terms of energy efficiency.

\section{Conclusion} \label{SecConc}
We consider the energy efficiency of compressed conventional collaborative sensing (CCCS) scheme focusing on balancing the tradeoff between energy efficiency and detection accuracy in cognitive radio environment. We first consider the existing CCCS scheme in the literature, and derive the achievable throughput, energy consumption and energy efficiency. The energy efficiency maximization for the CCCS scheme is posed as a non-convex, optimization problem. We approximated the optimization problem to reduce it to a convex optimization problem, and showed that this approximation holds with sufficient accuracy in the regime of interest. We analytically characterize the tradeoff between dimensionality reduction and collaborative sensing of CCCS scheme -- the implicit tradeoff between energy saving and detection accuracy, and show that by combining compression and collaboration the loss due to one can be compensated by the other which improves the overall energy efficiency of the cognitive radio network.

\section{Appendix}
\subsection{Proof of Theorem \ref{OptLambdaRan}} \label{OptLambdaRanNPThmProof}
To establish that $\pdcccs \geq \overline{P}_{d}$ is satisfied with equality, we show that $\frac{\partial \teecccs(\lambda, \comp, N)}{\partial \lambda} \geq 0$, for all $\lambda$. Observe that
\begin{align}
\frac{\partial \teecccs(\lambda, \comp, N)}{\partial \lambda} = \frac{\frac{\partial \trcccs(\lambda)}{\partial \lambda} \tecccs(\lambda)-\trcccs(\lambda)\frac{\partial\tecccs(\lambda)}{\partial \lambda}}{\tecccs^2(\lambda)}, \label{firstdriEE}
\end{align}
where
\begin{align}
\frac{\partial \trcccs(\lambda, \comp, N)}{\partial \lambda}=-\frac{\partial P_{f}}{\partial \lambda}(1+\phi)\phn \mathcal{C} (\ttot-\comp T_{s}), \label{firstdriR}
\end{align}
and
\begin{align}
\frac{\partial \tecccs(\lambda, \comp, N)}{\partial \lambda}=-\frac{\partial P_{f}}{\partial \lambda}\phn P_{t}(\ttot-\comp T_{s}). \label{firstdriE}
\end{align}
Upon further simplification, we get
\begin{align}
\frac{\partial \tecccs(\lambda, \comp, N)}{\partial \lambda} = -\frac{\partial P_{f}}{\partial \lambda} V_{1}(\lambda, \comp, N),
\end{align}
where
\begin{align}
& V_{1}(\lambda, \comp, N)=\left [\hspace{-0.1cm}  \frac{(1+\phi)\phn \mathcal{C} (\ttot-\comp T_{s}) \tecccs(\lambda, \comp, N)}{\tecccs^{2}(\lambda, \comp, N)} \right. \nonumber \\
& ~~~~~~~~~~~~~~~~~~ \left. - \frac{\phn P_{t}(\ttot-\comp T_{s})*\trcccs(\lambda, c, N)}{\tecccs^{2}(\lambda, \comp, N)}\right ]
\end{align}
Now, to show that $\frac{\partial \eecccs(\lambda, \comp, N)}{\partial \lambda} \geq 0$, it is enough to show that $V_{1}(\lambda, \comp, N) \geq 0$, since 
\begin{align*}
\frac{\partial P_{f}}{\partial \lambda}=-\frac{1}{2\varn\sqrt{\comp NP\pi}}\exp\left[-\frac{\left(\frac{\lambda}{\varn}-\comp NP\right)^2}{(4 \comp NP)}  \right ] \leq 0.
\end{align*}
In general, it is hard to analytically show that $V_1(\lambda, \comp, N) \geq 0$. However, since $\trcccs(\lambda, \comp, N) \geq 0$ and $\tecccs(\lambda, \comp, N) \geq 0$, the parameters $\phi, \mathcal{C}, \ttot$ and $T_s$ can be chosen such that $(1+\phi)\phn \mathcal{C} (\ttot-\comp T_{s}) \tecccs(\lambda, \comp, N) \geq \phn P_{t}(\ttot-\comp T_{s})\trcccs(\lambda, \comp, N)$. Later, in Sec.~\ref{SecResults}, it can be seen that the above condition is satisfied for those parameter values which are of practical interest. Therefore,
\begin{align}
& \overline{P}_d = Q\left(\frac{\frac{\lambda^*}{\vars+\varn} - \comp N P}{\sqrt{2 \comp N P}}\right) \nonumber \\
& \phantom{\overline{P}_d} = Q\left(\frac{\frac{\lambda^*}{\varn} \left( \frac{1}{1+\gamma} \right) - \comp N P}{\sqrt{2 \comp N P}}\right).
\end{align}
Rearranging the equation gives the expression for $\lambda^*$.

\subsection{Proof of Theorem \ref{cthmranNP}} \label{cthmranNPProof}
Note that
\begin{align}
\frac{\partial \teecccs(\lambda, \comp, N)}{\partial \comp} = \frac{\frac{\partial \trcccs(\comp)}{\partial \comp} \tecccs(\comp)-\trcccs(\comp)\frac{\partial\tecccs(\comp)}{\partial \comp}}{\tecccs^2(\comp)}. \label{firstdriEEwrtc}
\end{align} 
As $\comp \rightarrow 0$, it can be shown that
\begin{align}
\underset{\comp \rightarrow 0}{\lim}\frac{\partial \teecccs(\lambda, \comp, N)}{\partial \comp} \geq  \underset{\comp \rightarrow 0}{\lim} \left \{ - \frac{\partial P_{f}}{\partial \comp}  \left(\frac{\mathcal{C}}{P_t}\right) + V_{2}(\comp, N) \right\}, \label{EEdrivwrtcV}
\end{align}	
where
\begin{align} \label{eq:test38} 
& V_{2}(\comp, N) = \frac{\left [NP_s\tau_s+NP_t\tau_r\right]\phn\mathcal{C} }{P_t^2} \geq 0
\end{align} 
Also, note that
\begin{align}
&\frac{\partial P_f}{\partial  \comp}=-\frac{1}{\sqrt{\pi}}\exp\left(\frac{(\frac{\lambda^*}{\sigma_{w}^2}-\comp NP)^2}{4\comp NP}\right)\nonumber \\
&~~~~~~~~~~~~~~~~\left [-\frac{NP}{2\sqrt{\comp NP}}-\frac{NP((\frac{\lambda^*}{\sigma_{w}^2}-\comp NP))}{4(\comp NP^{3/2})}\right ] \label{dfdcwrtc} 	
\end{align}
Therefore, $P_f$ is a monotonically decreasing function of $\comp$. When $\comp \rightarrow  0$, it can be shown that $\frac{\partial P_f(\lambda, \comp, N)}{\partial  \comp} \rightarrow - \infty$. Since $V_2(\comp, N)$ is a positive constant, $\underset{\comp \rightarrow 0}{\lim}\frac{\partial \teecccs(\lambda, \comp, N)}{\partial \comp} = +\infty$. Furthermore, using a well-known bound on the Q function, we get the following lower bound $P_f$ as
\begin{align}
P_{f} \geq \left[1-\frac{2cNP}{(\frac{\lambda}{\sigma_w^2}-cNP)^2} \right] \exp{-\left[ \frac{(\frac{\lambda}{\sigma_w^2}-cNP)^2}{4cNP}\right ]}, \label{firstdir}
\end{align}
which can be used to get a lower bound on the first derivative of $\teecccs(\lambda, \comp, N)$ as
\begin{align}
&\frac{\partial \teecccs(\lambda, \comp, N)}{\partial \comp}\geq \underbrace{(BA-BD-BC-AE)}_{\triangleq X_1} \nonumber \\ &\hspace{-0.3cm}+\underbrace{(BC+AE-2BA+2BD)\left[1-\frac{2cNP}{(\frac{\lambda}{\sigma_w^2}-cNP)^2} \right] e^{-\left[ \frac{(\frac{\lambda}{\sigma_w^2}-cNP)^2}{4cNP}\right ]}}_{\triangleq X_2}\nonumber \\ &\hspace{-0.1cm}+\underbrace{(BA-BD)\left[1-\frac{2cNP}{(\frac{\lambda}{\sigma_w^2}-cNP)^2} \right]^2 e^{-\left[ \frac{(\frac{\lambda}{\sigma_w^2}-cNP)^2}{4cNP}\right ]}}_{\triangleq X_3} \nonumber \\ &-\underbrace{(AC) \frac{\partial P_f(\lambda, \comp, N)}{\partial \comp}}_{\triangleq X_4},
\end{align}
where
$A = \phn \mathcal{C} \left[\ttot- \comp T_s   \right] \geq 0$, $B =\phn P_t T_s \geq 0$, $C =NP_{s}\comp \tau_s + NP_t \comp \tau_r \geq 0$, and $D =P_t \left[\ttot- \comp T_s   \right] \phn \geq 0$.

As seen earlier, $\frac{\partial P_f}{\partial \comp}$ is negative, and it is easy to show that $BC+AE-2BA+2BD >0, BA-BD >0$, and consequently, $X_2 \geq 0$, $X_3 \geq 0$ and $X_4 \geq 0$. Now,
\begin{align} 
& \frac{\partial \teecccs(\lambda, \comp, N)}{\partial \comp} \geq X_1 +X_2+X_3+X_4 \nonumber \\
& ~~ \geq X_1 \nonumber \\
& ~~ = BA-BD-BC-AE \nonumber \\
& ~~ = (\phn^2 P_t T_s)\phn \mathcal{C} (\ttot-\comp T_s)-(\phn^2 P_t^2 T_s)(\ttot-\comp T_s)\nonumber \\
& ~~~~~~~~~ -(\phn P_t T_s)(NP_sc\tau_s+NP_tc\tau_r) \nonumber \\
& ~~~~~~~~~ -\hspace{-0.1cm}\phn\hspace{-0.01cm} \mathcal{C} (\ttot-\comp T_s) (NP_s\tau_s+\hspace{-0.1cm}NP_t \tau_r))\nonumber \\
& ~~ = \underbrace{\phn \left \{P_t T_s \mathcal{C}-P_t^2 T_s-\mathcal{C} (NP_s \tau_s+NP_t \tau_r) \right\}}_{\triangleq W} (\ttot \hspace{-0.1cm} - \hspace{-0.1cm} \comp T_s) \nonumber \\
& ~~~~~~~~~~ - \comp \underbrace{\phn P_tT_s(NP_s \tau_s+NP_t \tau_r)}_{\triangleq Y} \label{RqdFormEqn}
\end{align}
To ensure that $\frac{\partial \teecccs(\lambda, \comp, N)}{\partial \comp} \geq 0$, we need that the right hand side of \eqref{RqdFormEqn} to be $\geq 0$. Rearranging \eqref{RqdFormEqn}, observe that this is true when $\comp \leq \compupbnd \triangleq \frac{\ttot W}{T_s W+Y}$. In other words, we have shown that $\frac{\partial \teecccs(\lambda, \comp, N)}{\partial \comp} \geq 0$ whenever $\comp \in (0, \compupbnd)$. Finally, to establish that $\comp^* = \compbnd$, we need to show that $\compbnd \leq \compupbnd$. Although hard to show analytically, it is verified to be indeed true numerically, for moderate values of $N$ and for low SNR, which is of practical relevance.

\subsection{Proof of Theorem \ref{OptLambdaDetNPThm}} \label{OptLambdaDetNPThmProof}
Note that the first derivative of $P_f$ is negative as given below.
\begin{align}
\frac{\partial P_{f}}{\partial \lambda}=-\frac{1}{\varn\sqrt{2\pi \comp N \gamma}}\exp\left [\frac{-\lambda^2}{(2\comp N \gamma \sigma_{w}^{4})}  \right ] \leq 0
\end{align}
As mentioned earlier, since the expressions for average achievable throughput, average energy consumption and the energy efficiency expressions across all four scenarios $\textbf{S1}-\textbf{S4}$ for the deterministic case remains similar to the previous case, and similar set of arguments hold true for the deterministic case too. These can be used to prove that the first derivative of $\teecccs$ is greater than or equal to $0$. Therefore,
\begin{align}
& \overline{P}_d = Q\left(\frac{{\lambda^*} - \comp N \varn}{\varn\sqrt{ \comp N \gamma}}\right).
\end{align}
Rearranging the above equation gives the expression for $\lambda^*$.

\subsection{Proof of Theorem \ref{etathm}} \label{etathmProof}
Note that the first derivative of $\pfcccsdet$ with respect to $\comp$ from \eqref{pdpfequdet} is given by,
\begin{equation}  \label{firstdirpf1} 
\frac{\partial P_{f}(\lambda, N,\comp)}{\partial \comp}=-\frac{\exp(-(\frac{\lambda^2}{2N \gamma \sigma_{w}^{4}})N\gamma \lambda}{2\sqrt{2\pi}(cN\gamma)^{3/2} \varn} \leq  0.
\end{equation}
Therefore, $\underset{\comp\rightarrow 0}{\lim}\frac{\partial P_f(\lambda, \comp, N)}{\partial  \comp}\rightarrow - \infty$. Similar arguments given in Sec.~\ref{cthmranNPProof} can be used to show that $\comp^* = \compbnd$, even in this case.
 
\bibliographystyle{IEEEtran}
\bibliography{IEEEabrv,CCCbiblonew}
\end{document}